\documentclass[conference]{IEEEtran}

\usepackage{tikz}
\usetikzlibrary{shapes,arrows}
\usepackage{subcaption}
\usepackage{float}
\usepackage{epsfig}
\usepackage{xspace}
\usepackage{color}
\usepackage{pgfplots}
\usetikzlibrary{matrix}
\usetikzlibrary{patterns}
\usepgfplotslibrary{groupplots}
\pgfplotsset{compat=newest}

\usepackage{url}
\usepackage{algpseudocode}
\usepackage{amsthm}
\usepackage{framed}

\newcommand{\tool}{{\sc Ohie}\xspace}
\newcommand{\codename}{{\sc Ohie}\xspace}
\newcommand{\blocksize}{$20$ KB\xspace}
\newcommand{\Paragraph}[1]{\smallskip\noindent{\bf #1}}
\newtheorem{theorem}{Theorem}

\newtheorem{lemma}[theorem]{Lemma}

\renewcommand{\paragraph}[1]{\pagebreak\mbox{XXXXX} \pagebreak\mbox{XXXXX} \pagebreak\mbox{XXXXXX}}

\usepackage{cite}
\usepackage{amsmath,amssymb,amsfonts}
\usepackage{graphicx}
\usepackage{textcomp}
\usepackage{xcolor}
\def\BibTeX{{\rm B\kern-.05em{\sc i\kern-.025em b}\kern-.08em
    T\kern-.1667em\lower.7ex\hbox{E}\kern-.125emX}}

\begin{document}

\title{\Huge \codename: Blockchain Scaling Made Simple}
\author{
{ Haifeng Yu \hspace*{18mm} Ivica Nikoli\'c \hspace*{18mm} Ruomu Hou \hspace*{18mm} Prateek Saxena}\\
{Department of Computer Science, National University of Singapore}\\
{\em haifeng@comp.nus.edu.sg \,
  inikolic@nus.edu.sg \, houruomu@comp.nus.edu.sg \, prateeks@comp.nus.edu.sg}
}

\maketitle

\begin{abstract}
Many blockchain consensus protocols have been proposed
recently to {\em scale} the throughput of a blockchain with
available bandwidth. However,
these protocols are becoming increasingly complex, making it more and more difficult to
produce proofs of their security guarantees.
We propose a novel permissionless blockchain protocol \codename which
explicitly aims for {\em simplicity}. \codename composes as many
parallel instances of Bitcoin's original (and simple) backbone
protocol as needed to achieve excellent throughput.
We formally prove the safety and liveness properties of \codename. We
demonstrate its performance with a prototype implementation and
large-scale experiments with up to $50,000$ nodes.
In our experiments, \codename achieves linear
scaling with available bandwidth, providing about $4$-$10$Mbps
transaction throughput (under $8$-$20$Mbps per-node available
bandwidth configurations) and at least about $20$x better
decentralization over prior works.

\end{abstract}

\section{Introduction}
\label{sec:intro}

Blockchain protocols power several open computational platforms, which
allow a network of nodes to agree on the state of a distributed ledger
periodically. The distributed ledger provides a total order over
{\em transactions}. Nodes connect to each other via an open
peer-to-peer (P2P) overlay network, which is {\em permissionless}: it
allows any computer to connect and participate in the computational
service without registering its identity with a central authority. A
blockchain consensus protocol enables every node to {\em confirm} a
set of transactions periodically, batched into chunks called blocks.
The protocol ensures that honest nodes all
agree on the same total ordering of the confirmed blocks, and that
the set of confirmed blocks grow over time.  The
earliest such protocol was in Bitcoin~\cite{bitcoin}, and has
spurred interest in many blockchain platforms since.




%
%
%
%

The seminal Bitcoin protocol, published about a decade ago~\cite{bitcoin}, laid some
of the key foundations for modern blockchain protocols.
But as Bitcoin gained popularity, its low throughput has been cited as
glaring concern resulting in high costs per
transaction~\cite{scaling16}. Currently, Ethereum and Bitcoin process
only about $5$KB or $10$ transactions per second on average, which is less than $0.2\%$ of the average available bandwidth in their respective
P2P networks~\cite{gencer-fc18}.  Many recent
research efforts have thus focused on improving the transaction
throughput, resulting in a series of beautiful designs for
permissionless
blockchains~\cite{algorand,elastico,omniledger,bitcoin-ng,byzcoin,rapidchain,solida,peercensus,phantom,conflux}.

Bitcoin's core consensus protocol---called Nakamoto consensus---still stands out in one critical aspect: it is remarkably {\em simple}. Nakamoto consensus can be fully described in a few dozens of lines
of pseudo-code. Such simplicity makes it extensively
amenable to re-parameterization in hundreds of deployments, and more
importantly, a series of formal proofs on
its security guarantees have been carefully
and independently established in several research
works~\cite{garay2015bitcoin,pow-difficulty,kiayias15tradeoff,kiffer18,pass2017analysis}.

The importance of keeping constructions {\em simple} enough to
allow such formal proofs and cross validations cannot be
over-emphasized.  Consensus protocols are notoriously difficult to
analyze in the presence of byzantine failures.
Formal proofs/analysis are especially
important to protocols that are difficult to upgrade once deployed:
Upgrades of blockchain protocols after deployment (i.e., ``hard forks'') cause both
philosophical disagreements and financial impact.

Some recent high-throughput blockchain protocols do strive to retain the
simplicity of Nakamoto consensus. Unfortunately, many of them do not come with formal end-to-end security proofs. As an example,
Conflux~\cite{conflux} is a recent high-throughput blockchain protocol with an elegant design.
But it has only provided
informal security
arguments (Section 3.3 in \cite{conflux}).
Appendix~\ref{sec:conflux}
shows that in our simulation, as the throughput of Conflux increases,
the security properties of Conflux deteriorate.\footnote{Related
  observations on Conflux have also been independently made by Bagaria
  et al.~\cite{bagaria18} and Fitzi et al.~\cite{fitzi18}.}  Such an
undesirable property of Conflux is hard to discover via informal
arguments. The Conflux paper itself also presented effective attacks on a
prior protocol (Phantom~\cite{phantom}) that comes without proofs.

\Paragraph{Our goal.}
This work aims to develop a {\em simple} blockchain consensus protocol, which should admit
formal end-to-end proofs on safety and liveness, while
retaining the high throughput achieved by
state-of-the-art blockchain protocols. Specifically, we aim to achieve:

\begin{enumerate}
\item {\em Near-optimal resilience}: Tolerate an adversarial computational power fraction $f$ close to $\frac{1}{2}$, which is near-optimal;

\item {\em Throughput approaching a significant fraction of the raw network bandwidth}:
The raw available network bandwidth in the P2P network constitutes a crude throughput upper bound for all blockchain protocols. We aim to achieve a throughput, in terms of transactions processed per second, that approaches a significant fraction of this raw network bandwidth.\footnote{Note that the raw available bandwidth is only a rather {\em crude} upper bound, and hence in practice it is unlikely for the throughput to reach this upper bound. For example, this crude upper bound does not take into account factors such that TCP slow start, probabilistic block generation, probabilistic hot-spots in the P2P overlay network, overheads for determining which blocks to gossip, and so on.}

\item {\em Decentralization}: Many dynamically selected
  block proposers should be able to add blocks to the
  blockchain per second, rather than for example, having one leader or a small
  committee add blocks over a long period of time. More block proposers make
  transactions less susceptible to censorship~\cite{smartpool,censorcoin}, and a
  DoS attack against a small number of nodes will no longer impact the
  availability of the entire system~\cite{heilman2015eclipse}.
\end{enumerate}


\Paragraph{Our approach.}
This work proposes \tool,\footnote{The word ``ohie'' comes from the Maori language and means ``simple''.} a novel
blockchain protocol for the permissionless setting.
Specifically, \tool composes many {\em parallel} instances of the
Nakamoto consensus protocol. \tool first applies a simple mechanism to
force the adversary to evenly split its adversarial computational
power across all these chains (i.e., Nakamoto consensus
instances). Next, \tool
proposes a simple solution to securely
arrive at a global order for blocks across all the parallel chains,
hence achieving consistency.


The {\em modularity} of \tool enables us to prove (under any given constant $f< \frac{1}{2}$)
its safety and liveness properties via a {\em reduction} from those of Nakamoto consensus.
Our proof
invokes existing theorems on Nakamoto
consensus, making the proof modular and
streamlined~\cite{pass2017analysis}.
By running as many (e.g., $1000$) chains  as the
network bandwidth permits,
\tool's throughput scales with available network bandwidth. Finally, the parallel chains in \tool lead to excellent decentralization, since many miners can simultaneously add new blocks.

\Paragraph{Our results.}
We have implemented a prototype of \tool,
the source code of which is publicly available~\cite{ohie-code}.
We have evaluated it on Amazon EC2 with up
to $50,000$ nodes, under similar settings as in prior
works~\cite{algorand,conflux}. Our evaluation first shows that
\codename's throughput scales linearly with
available bandwidth, as is the case with state-of-the-art
protocols~\cite{rapidchain,elastico,omniledger,bitcoin-ng,algorand}.
For example, under configurations with $8$-$20$Mbps per-node
bandwidth, \codename achieves about $4$-$10$Mbps transaction throughput. This translates to close to $1000$ to $2500$ transactions per
second, assuming $500$-byte average transaction size as in Bitcoin. Such throughput is about $550\%$
of the throughput of AlgoRand~\cite{algorand} and $150\%$ of the
throughput of Conflux~\cite{conflux} under
similar available bandwidth.  This suggests that while explicitly focusing on simplicity, \tool retains the high throughput property of modern blockchain designs. Second, regardless of the throughput, the
confirmation latency for blocks in \tool is always below $10$ minutes in our experiments, under security parameters comparable to Bitcoin and Ethereum
deployments. (The confirmation latencies in
Bitcoin and Ethereum are $60$ minutes and $3$ minutes, respectively.)
Finally, our experiments show that the decentralization factor of \codename is at least about $20$x of all prior works.

\section{System Model and Problem}
\label{sec:model}

\Paragraph{System model.}
Our system model and assumptions directly follow several prior works (e.g., \cite{kiffer18,pass2017analysis}).
We model hash functions as random oracles, and assume that some random {\em genesis blocks} are available from
an initial trusted setup. 
We consider a permissionless setting, where nodes have no pre-established identities. We use standard proof-of-work (PoW) puzzles, a form of sybil resistance, to limit the adversary by computation
power. We assume that the entire network has total $n$ units of computational power, and some reasonable estimation of $n$ is known. Out of this, the {\em adversary} controls $fn$ units of computational power, with $f$ being any constant below $\frac{1}{2}$. The adversary can deviate arbitrarily from the prescribed protocol, and hence is {\em byzantine}. We assume that some procedure to estimate the total computation
power exists a-priori~\cite{pow-difficulty}.  Standard PoW schemes help ascertain this periodically. For instance, in
Bitcoin, the rate of PoW solutions is adjusted (periodically) to be approximately
$10$ minutes, and the PoW difficulty essentially maps to the estimated total computation
power in the network. We can use the same mechanism in our design.


Given a fixed block size (e.g., \blocksize), we assume that honest nodes form a well-connected synchronous P2P overlay network, so that an honest node can broadcast (via gossiping) such a block with a maximum latency of $\delta$ to other honest nodes.
Our protocol, much like Bitcoin, can tolerate
variations in the actual propagation delay.
%
Network partitioning attacks can delay block delivery arbitrarily and
can cause honest nodes to lose
inter-connectivity~\cite{hijacking-bitcoin}. Defences to mitigate
these attacks are an important area of research; however, they are
outside the scope of the design of the consensus protocol. If the
network becomes completely asynchronous, blockchain consensus is
considered impossible~\cite{pass2017analysis}. In the presence of
partitions, the CAP theorem suggests that protocols can either choose
liveness or safety, but not
both~\cite{cap-theorem-perspectives,cap-theorem}; we choose
liveness---the same as Nakamoto consensus.
The adversary sees every message as soon as it is sent.
The adversary can arbitrarily inject its own messages into the system
at any time (this captures the {\em selfish mining}
attack~\cite{eyal2018majority}, where newly mined blocks are injected
at strategic points of time).



\Paragraph{Problem definition.}
A blockchain protocol should enable any node at any time to output a sequence of total-ordered blocks, which we call the {\em sequence of confirmed blocks} (or {\sc SCB} in short). For example, in the Bitcoin protocol, the SCB is simply the blockchain itself after removing the last $6$ blocks. {\em Safety} and {\em liveness}, in the context of blockchain protocols, correspond to the {\bf consistency} and {\bf quality-growth} properties of the SCB. Informally, these two properties mean that the
SCB's on different honest nodes at different time are always consistent with each other, and that the total number of {\em honest blocks} (i.e., blocks generate by honest nodes) in the SCB grows over time at a healthy rate. We leave the formal definitions of these properties to Section~\ref{sec:analysis}.


Having {\bf consistency} and {\bf quality-growth}
is sufficient to enable a wide range of different
applications. For example, {\bf consistency} prevents double-spending in a cryptocurrency---if two
transactions spend the same coin, all nodes honor only the first
transaction in the total order
\footnote{Transactions spending the same coin as earlier ones in the
  total order can simply be skipped as invalid by the user.}.
Similarly, any ``conflicting'' state updates by smart contracts
running on the blockchain can be ordered consistently by all nodes, by
following the ordering in the SCB. In fact, ensuring a total order
is key to support many different consistency
properties~\cite{sequential-consistency,eventual-consistency-bailis,concurrency-db-79}
for applications, including for smart
contracts~\cite{chainspace}.

Finally, if the same transaction is included in multiple blocks in the SCB, then the first occurrence will be processed, while the remaining occurrences will be ignored. To avoid such waste, miners should ideally pick different transactions to
include in blocks.
For example, among all transactions, a miner can pick those transactions whose hashes are the closest to the hash of the miner's public key. Note that
this does not impact safety or liveness at all; it simply reduces the possibility of multiple inclusions (in different blocks) of the same
transaction.

\section{Conceptual Design}
\label{sec:conceptual}

\tool composes $k$ (e.g.,
$k=1000$) {\em parallel} instances of Nakamoto consensus.
Intuitively, we also call these $k$ parallel instances as $k$ parallel
chains. Each chain has a distinct genesis block, and the chains have
ids from $0$ to $k-1$ (which can come from the lexicographic order of all the genesis
  blocks).
Within each instance, we follow the longest-path-rule in
Nakamoto consensus.\footnote{Nakamoto consensus, strictly speaking, maintains a tree of blocks~\cite{bitcoin}. The longest-path-rule selects the longest path from the root (i.e., genesis block)
  to some leaf of the tree, where the leaf is chosen such that the
  path length is maximized. This path is often referred to
  as the chain.}
The miners in \tool extend the $k$ chains
concurrently.

\begin{figure*}
  \centering
  \IfFileExists{naivedesign.pdf}{\includegraphics[width=15cm]{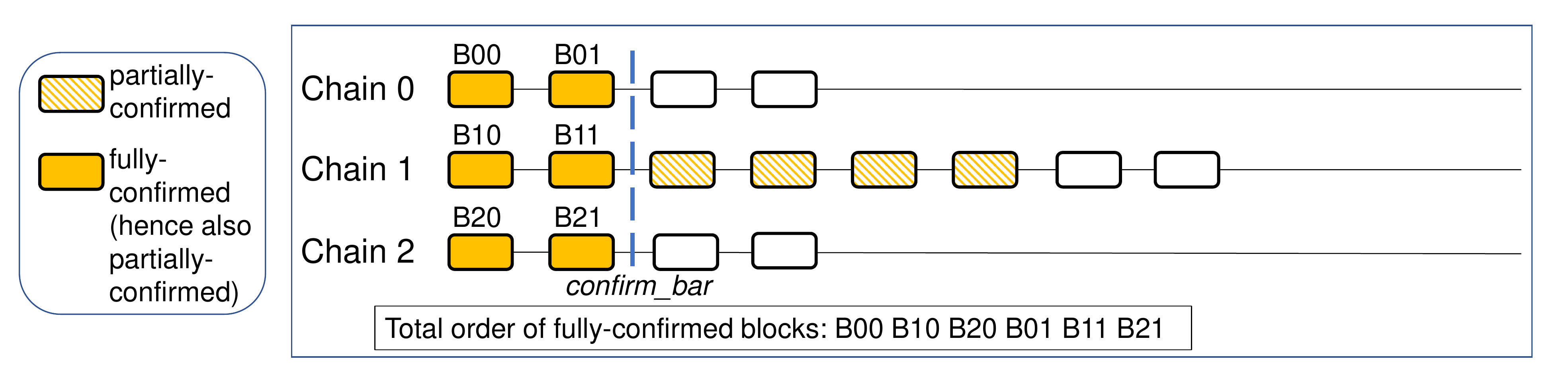}}{\includegraphics[width=16cm]{blockchainpar/naivedesign}}
  \vspace{-0mm}
  \caption{Illustrating ${\tt confirm\_bar}$ under $T=2$. Here chain $0$, $1$, and $2$ have $2$, $6$, and $2$ partially-confirmed blocks, respectively. On chain $1$, only the first $2$ blocks are fully-confirmed.}
  \vspace{-1mm}
  \label{fig:naivedesign}
\end{figure*}

\subsection{Mining in \tool}

Consider any fixed block size (e.g., \blocksize), and the corresponding block propagation delay $\delta$ (e.g., $2$ seconds).
Existing results~\cite{kiffer18,pass2017analysis} on Nakamoto consensus show that for any given constant $f<\frac{1}{2}$, there exists some constant $c$ such that if the {\em block interval} (i.e., average time needed to generate the next block on the chain) is at least $c\cdot \delta$, then Nakamoto consensus will offer some nice security properties. (Theorem~\ref{the:blackbox} later makes this precise.) As an example, for $f =0.43$, it suffices~\cite{pass2017analysis} for $c=5$. Such $c\cdot \delta$, together with $n$, then maps to a certain {\em PoW difficulty} $p$ in Nakamoto consensus, which is the probability of mining success for one hash operation. (Theorem~\ref{the:blackbox} gives the precise mapping.) Given $p$, the hash of a valid block in Nakamoto should have $\log_2 \frac{1}{p}$ leading zeros.

In \tool, to tolerate the same $f$ as above, the hash of a valid block should have $\log_2 \frac{1}{kp}$ leading zeros, where the value $p$ is chosen to be the same as in the above Nakamoto consensus protocol.
Next, the last $\log_2 k$ bits of the hash\footnote{Under our random oracle assumption, any $\log_2 k$ bits of the hash (other that the first $\log_2 \frac{1}{kp}$ bits) work. In implementation, one can choose to use any portion (other that the first $\log_2 \frac{1}{kp}$ bits) of the hash as appropriate, based on the specific hash function.} of the \tool block will index to one of the $k$ chains in \tool, and the block will be assigned to and will extend from that chain.
We have assumed the hash function to be a random oracle, and note that $\log_2 \frac{1}{kp} + \log_2 k = \log_2\frac{1}{p}$. Hence for any {\em given} chain in \tool, the probability of one hash operation ({\em either} done by honest nodes {\em or} done by the adversary) generating a block for that chain\footnote{Namely, total $\log_2\frac{1}{p}$ positions in the block's hash must match some pre-determined values, respectively.} is exactly $p$, which is the same as in Nakamoto consensus.
Similarly, the block interval for any given chain in \tool will be the same as in Nakamoto consensus.

The above relation can be formalized: Taking all mechanisms in \tool (especially the Merkle tree mechanism described next) into account, Lemma~\ref{lemma:reduction} later will prove that the behavior of any given chain in \tool almost follows exactly the  same distribution as the behavior of the single chain in Nakamoto consensus. Note that different chains in \tool are still {\em correlated}, since a block is assigned to exactly one chain. But we will be able to properly bound the probability of all bad events (whether correlated or not) via a simple union bound.

Finally, since \tool has $k$ parallel chains, on expectation there will be total $k$ blocks (across all chains) generated every $c \cdot \delta$ time, instead of just one block. Our experiments later will confirm the following simple yet critical property: {\em Propagating many parallel blocks has minimal negative impact on the block propagation delay $\delta$, as compared to propagating a single such block, until we start to saturate the network bandwidth of the system.} Hence in \tool, we
use as large a $k$ as possible to effectively utilize all the bandwidth in the system, subject to the condition that $\delta$ is minimally impacted.

\subsection{Security of Individual Chains}
\label{sec:individualchain}

In Nakamoto consensus, a new block $B$ extends from some existing block $A$. The
PoW computes over $B$, which contains the hash
of $A$ as a field, cryptographically binding the extension of $A$ by
$B$.
In \tool, however, a miner does not know
which chain a new block will extend until it finishes the PoW puzzle,
the last $\log_2 k$ bits of which then determine the chain extended.

To deal with this, in \tool, a miner uses a  Merkle tree~\cite{merkle-tree} to bind to the last
blocks of all the $k$
chains in its local view.  Specifically, let $A_i$ be the last block\footnote{Exactly the same as in Nakamoto consensus, the last block here refers to the very last block on the longest path from the genesis block.
In particular, this last block is not yet partially-confirmed. (In fact, none of the last $T$ blocks on the path are partially-confirmed.)} of chain $i$, for $0\le i\le k-1$. The miner computes a Merkle tree using $hash(A_0)$ through $hash(A_{k-1})$ as the tree leaves. The root of the Merkle tree
is included in the new block $B$ as an input to the PoW puzzle. After
$B$ is mined, the integer $i$ that corresponds to the last $\log_2 k$
bits of $hash(B)$ determines the block $A_i$ from which $B$ extends.
When disseminating $B$ in the network, a miner includes $hash(A_i)$
and the Merkle proof of $hash(A_i)$ in the message. (The value of $i$ will be directly obtained from the hash of $B$.) The Merkle proofs
are standard, consisting of $\log_2 k$ off-path hashes~\cite{merkle-tree}.

Intuitively, the above design binds each successful PoW to a {\em
  single} existing block on a {\em single} chain from which the new
block extends. For further understanding, let us consider the
following example scenario.  The adversary may intentionally choose
$A_0$ through $A_{k-1}$ all from (say) chain $3$, for constructing the
Merkle tree. Assume the adversary finds a block $B$ whose hash has
$\log_2\frac{1}{kp}$ leading zeros. If the last $\log_2 k$ bits of $B$
does not equal to $3$, then $B$ will not be accepted by any honest
node. Otherwise $B$ will be accepted, and $B$ can only extend from
$A_3$ (instead of any other block) on chain $3$. The reason is that
the honest nodes will need to verify the $4$th leaf (which corresponds
to chain $3$) on the Merkle tree. Only $A_3$ can pass such
verification. Also note that the adversary may intentionally not use
the last block on chain $3$ as $A_3$. This is not a problem, since the
security of \tool ultimately inherits (see Section~\ref{sec:analysis})
from the security of Nakamoto consensus~\cite{pass2017analysis}, and
the adversary in Nakamoto consensus can already extend from any block
(instead of extending from the last block).

\begin{figure*}
  \centering
  \IfFileExists{growth.pdf}{\includegraphics[width=16cm]{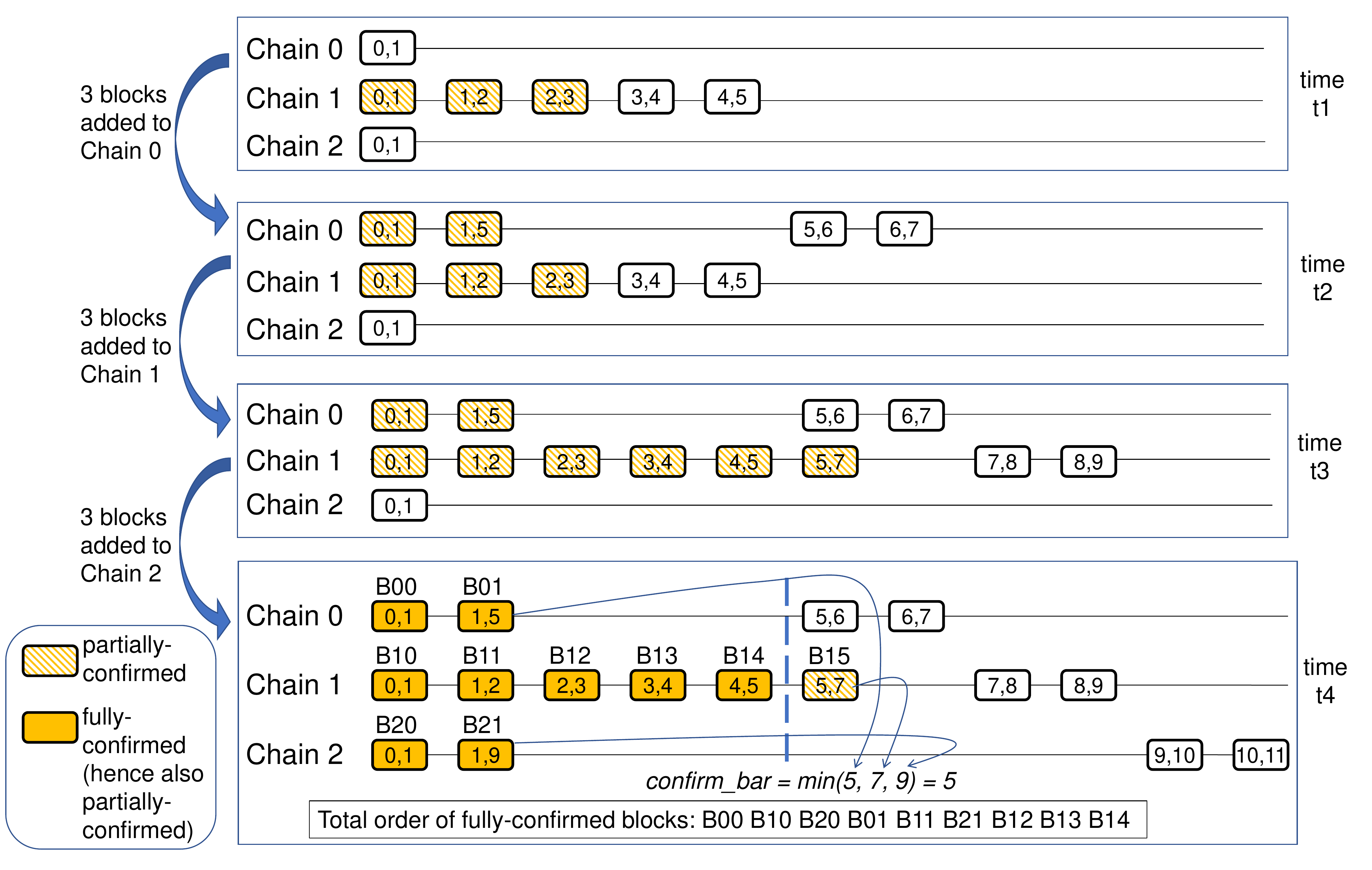}}{\includegraphics[width=16.5cm]{blockchainpar/growth}}
  \vspace{-0mm}
  \caption{Each block has a tuple $({\tt rank},\,\, {\tt next\_rank})$. Our example here assumes that once a block is generated, it is seen by all nodes immediately. Note that \tool does not need such an assumption. We use $T=2$ in this figure.}
  \vspace{-2mm}
  \label{fig:growth}
\end{figure*}

\subsection{Ordering Blocks across Chains -- A Starting Point}
\label{sec:totalorder}

Section~\ref{sec:analysis} will show that each individual chain in \tool
inherits the proven security
properties of Nakamoto consensus~\cite{pass2017analysis}. For
example, with high probability, all blocks on a chain except the last
$T$ blocks (for some parameter $T$) are {\em confirmed} --- the ordering
of these confirmed blocks on the chain will no longer change in the
future. This however does not yet give us a total ordering of all the
confirmed blocks across all the $k$ chains in \tool.
Recall from Section~\ref{sec:model} that a node needs to
generate an SCB (i.e., a total order of all confirmed blocks)
satisfying {\bf consistency} and {\bf quality-growth}.
To avoid notational collision, from this point on, we call all
blocks on a chain except the last $T$ blocks as {\em
partially-confirmed}. Once a partially-confirmed block is added to
SCB, it becomes {\em fully-confirmed}.


%

One way to design the SCB is to first include the
first partially-confirmed block on each of the $k$ chains (there are total
$k$ such blocks, and we order them by their chain ids), and then add
the second partially-confirmed block on each of the $k$ chains, and so
on. This would work well, if every chain has the same number of
partially-confirmed blocks.

When the chains do not have the same number of partially-confirmed
blocks, we will need to impose a {\em confirmation bar} (denoted as
${\tt confirm\_bar}$) that is limited by the chain with the smallest
number of partially-confirmed blocks (see
Figure~\ref{fig:naivedesign}). Blocks after ${\tt confirm\_bar}$ cannot
be included in the total order yet.
This causes a serious problem, since with blocks extending chains at
random, some chains can have more blocks than others. In fact in our experiments (results not shown), such imbalance appears to even grow unbounded over time.


%


\subsection{Ordering Blocks across Chains -- Our Approach}
\label{sec:ourorder}

Imagine that the longest chain is $8$ blocks longer than the shortest
chain. Our basic idea to overcome the previous problem
is that when the next block on the shortest
chain is generated, we simply view it as $8$ blocks worth.
%
Figure~\ref{fig:growth} illustrates this idea. Here
each block has two additional fields used for ordering blocks across chains,
denoted as a tuple $({\tt rank},\,\, {\tt next\_rank})$. In the total
ordering of fully-confirmed blocks, the blocks are ordered by
increasing ${\tt rank}$ values, with tie-breaking based on the {\em
  chain ids}. The {\em chain id} of a block is simply the id of the
chain to which the block belongs. For any new block $B$ that extends
from some existing block $A$, we directly set $B$'s
${\tt rank}$ to be the same as $A$'s ${\tt next\_rank}$. Putting it another way, $A$'s ${\tt next\_rank}$ specifies (and fixes) the ${\tt rank}$ of $B$. A genesis block always has ${\tt rank}$ of $0$ and ${\tt next\_rank}$ of $1$.

\Paragraph{Determining ${\tt next\_rank}$.}
Properly setting the ${\tt next\_rank}$ of a new block $B$ is key to our
design. A miner sees all the chains, and can
infer the expected ${\tt rank}$ of the next upcoming block on each chain
(before $B$ is added to its chain). For example, at time $t_1$
in Figure~\ref{fig:growth}, the ${\tt next\_rank}$
of the current last block (not the last {\em partially-confirmed} block) on each of the three chains is $1$, $5$, and
$1$, respectively. Hence these will be the ${\tt rank}$ of the
upcoming blocks on those chain, respectively. Let $x$ denote the
maximum (i.e., $5$) among these values, and
$x$ corresponds to the ``longest'' chain (in terms of rank)
among the $k$ chains. Regardless of which chain the new block $B$ ends
up belonging to, we want $B$ to help that chain to increase its rank to
catch up with the ``longest'' chain. Hence the node generating $B$ should directly set $B$'s ${\tt next\_rank}$ to be $x$, or any value larger than\footnote{Using a value larger than $x$ will cause $B$'s chain to exceed the length of the currently ``longest'' chain (in terms of rank). This is not a problem since the other chains will catch up with $B$'s chain once a new honest block is added to each of those chains. We do not need the chains to have exactly the same length all the time.} $x$.
(To prevent adversarial manipulation, a careful implementation requiring an additional
${\tt trailing}$ field is needed---
see Section~\ref{sec:full}.) Finally, we always ensure that $B$'s ${\tt next\_rank}$ is at least one larger than $B$'s ${\tt rank}$, regardless of $x$. This guarantees that the ${\tt rank}$ values of blocks on one chain are always increasing.

%

In the example in Figure~\ref{fig:growth}, from time $t_1$ to $t_2$,
there are 3 new blocks (with tuples $(1,5)$, $(5,6)$, and $(6,7)$, respectively) added to chain $0$. For the first new block added (i.e., the block with tuple $(1,5)$), the value of $x$ is $5$, and hence the block's ${\tt next\_rank}$ is set to be $5$. For the
second new block (i.e., the block with tuple $(5,6)$), the value of $x$ is still $5$, while the ${\tt rank}$ of this block
is already $5$. Hence we set the block's
${\tt next\_rank}$ to $6$.


\Paragraph{Determining the total order.}
We can now establish a total order among the blocks in the following
way. Consider any given honest node at any given time and its local
view of all the chains. Let $y_i$ be the ${\tt next\_rank}$ of the last
partially-confirmed block on chain $i$ in this view. For example, at
time $t_4$ in Figure~\ref{fig:growth}, we have $y_0 = 5$, $y_1=7$ and
$y_2 = 9$. Note that the position of a partially-confirmed block on
its respective chain will not change anymore, and hence all these $y_i$'s are
``stable''. Let ${\tt confirm\_bar} \leftarrow \min_{i=1}^k
y_i$. Then, the next partially-confirmed block on any chain
must have a ${\tt rank}$ no smaller than ${\tt confirm\_bar}$.
This means that the node must have seen all
partially-confirmed blocks whose ${\tt rank}$ is smaller than
${\tt confirm\_bar}$. Thus, it is safe (see Lemma~\ref{lemma:carryover}) to deem all
partially-confirmed blocks whose rank is smaller than ${\tt confirm\_bar}$ as
fully-confirmed, and include them in SCB. Finally, all the fully-confirmed blocks will be ordered by increasing ${\tt rank}$ values, with tie-breaking favoring smaller chain ids.
%
As an example, in Figure~\ref{fig:growth}, at time $t_4$, we have ${\tt
  confirm\_bar}$ being $5$. Hence, the $9$ partially-confirmed blocks
whose ${\tt rank}$ is below $5$ become fully-confirmed.

\Paragraph{Summary.}
By properly setting the ${\tt rank}$ values of all the blocks, we
 ensure that the chains remain balanced in terms of the ${\tt
rank}$'s of their respective last blocks. This is regardless of
how imbalanced the chains are in terms of the total number of blocks, hence avoiding the earlier imbalance problem in Section~\ref{sec:totalorder}.


\begin{figure*}
\begin{subfigure}[t]{.4\textwidth}
\pgfplotsset{compat=newest,
legend style={font=\footnotesize},
label style={font=\footnotesize},
tick label style={font=\footnotesize},
title style={font=\footnotesize}}
\centering\captionsetup{width=.8\linewidth}
\begin{framed}
\footnotesize\small
\begin{algorithmic}[1]
\State $d \leftarrow \log_2 (\frac{1}{kp})$; 
\State $V_i \leftarrow$ \{(genesis block of chain $i$, the attachment for genesis block of chain $i$)\}, for $0\le i\le k-1$;
\State $M \leftarrow $ Merkle tree of the hashes of the $k$ genesis blocks;
\label{code:h2computeinitial}
\State ${\tt trailing}$ $\leftarrow $ hash of genesis block of chain $0$;
\State
\State {\sc Ohie}() \{
\State \hspace*{2mm}{\bf repeat forever} \{
\State \hspace*{6mm}ReceiveState(); \label{code:mainloop1}
\State \hspace*{6mm}Mining();
\State \hspace*{6mm}SendState(); \label{code:mainloop2}
\State \hspace*{2mm}\}
\State \}
\State
\State Mining() \{
\State \hspace*{2mm}$B.{\tt transactions} \leftarrow$ get\_transactions();
\State \hspace*{2mm}$B.{\tt root}$ $\leftarrow$ root of Merkle tree $M$;
\State \hspace*{2mm}$B.{\tt trailing} \leftarrow$ ${\tt trailing}$;
\State \hspace*{2mm}$B.{\tt nonce} \leftarrow$  new\_nonce();
\State \hspace*{2mm}$\widehat{B}.{\tt hash} \leftarrow hash(B)$;
\label{code:h1compute}
\State \hspace*{2mm}{\bf if} ($\widehat{B}.{\tt hash}$ has $d$ leading zeroes) \{
\State \hspace*{6mm}$i$ $\leftarrow$ last $\log_2 k$ bits of $\widehat{B}.{\tt hash}$;
\State \hspace*{6mm}$\widehat{B}.{\tt leaf} \leftarrow$ leaf $i$ of $M$;
\State \hspace*{6mm}$\widehat{B}.{\tt leaf\_proof} \leftarrow M$.MerkleProof($i$);
\State \hspace*{6mm}ProcessBlock($B$, $\widehat{B}$);
\State \hspace*{2mm}\}
\State \}
\State
\State SendState() \{
\State \hspace*{2mm} send $V_i$ ($0\le i\le k-1$) to other nodes;
\State \hspace*{2mm} // In implementation, only need to send those blocks not sent before.
\State \}
\State \State ReceiveState() \{
\State \hspace*{2mm} {\bf foreach} $(B, \widehat{B})\in$ received state {\bf do}
\State \hspace*{6mm} ProcessBlock($B$, $\widehat{B}$);
\State // A block $B_1$ should be processed before $B_2$, if $\widehat{B_2}.{\tt leaf}$ or $\widehat{B_2}.{\tt trailing}$ points to $B_1$.
\State \}
\end{algorithmic}
\normalsize
\end{framed}
\end{subfigure}
\begin{subfigure}[t]{.59\textwidth}
\pgfplotsset{compat=newest,
legend style={font=\footnotesize},
label style={font=\footnotesize},
tick label style={font=\footnotesize},
title style={font=\footnotesize}}
\centering\captionsetup{width=.9\linewidth}
\begin{framed}
\footnotesize\small
\begin{algorithmic}[1]
\setcounter{ALG@line}{37}
\State ProcessBlock($B$, $\widehat{B}$) \{
\State \hspace*{2mm}// do some verifications
\State \hspace*{2mm}$i$ $\leftarrow$ last $\log_2 k$ bits of $\widehat{B}.{\tt hash}$;
\State \hspace*{2mm}verify that $\widehat{B}.{\tt hash}$ has $d$ leading zeroes;
\State \hspace*{2mm}verify that $hash(B) = \widehat{B}.{\tt hash}$;
\label{code:h1verify}
\State \hspace*{2mm}verify that $\widehat{B}.{\tt leaf}$ is leaf $i$ in the Merkle tree, based on $B.{\tt root}$ and $\widehat{B}.{\tt leaf\_proof}$;
\label{code:h2verify}
\State \hspace*{2mm}verify that $\widehat{B}.{\tt leaf} = \widehat{A}.{\tt hash}$ for some block $(A, \widehat{A}) \in V_i$;
\State \hspace*{2mm}verify that $B.{\tt trailing} = \widehat{C}.{\tt hash}$ for some
block $(C, \widehat{C}) \in \cup_{j=0}^{k-1}V_j$;
\State \hspace*{2mm}{\bf if} (any of the above 5 verifications fail) {\bf then return};
\State
\State \hspace*{2mm}// compute ${\tt rank}$ and ${\tt next\_rank}$ values
\State \hspace*{2mm}$\widehat{B}.{\tt rank} \leftarrow \widehat{A}.{\tt next\_rank}$;
\label{code:rankstart}
\State \hspace*{2mm}$\widehat{B}.{\tt next\_rank} \leftarrow \widehat{C}.{\tt next\_rank}$;
\State \hspace*{2mm}{\bf if} ($\widehat{B}.{\tt next\_rank}\le \widehat{B}.{\tt rank}$) {\bf then} $\widehat{B}.{\tt next\_rank}\leftarrow \widehat{B}.{\tt rank}+1$;
\label{code:rankend}
\State
\State \hspace*{2mm}// update local data structures
\State \hspace*{2mm}$V_i \leftarrow V_i \cup \{(B, \widehat{B})\}$;
\State \hspace*{2mm}update ${\tt trailing}$;
\State \hspace*{2mm}update Merkle tree $M$;
\label{code:h2compute}
\State \}
\State
      \State OutputSCB() \{
      \State \hspace*{2mm}// determine partially-confirmed blocks and ${\tt confirm\_bar}$
      \State \hspace*{2mm}{\bf for} ($i=0$; $i < k$; $i$++) \{
      \State \hspace*{6mm}$W_i$ $\leftarrow$ get\_longest\_path($V_i$);
      \State \hspace*{6mm}${\tt partial}_i$ $\leftarrow$ blocks in $W_i$ except the last $T$ blocks;
      \State \hspace*{6mm}($B_i$, $\widehat{B_i}$)  $\leftarrow$ the last block in ${\tt partial}_i$;
      \State \hspace*{6mm}$y_i\leftarrow \widehat{B_i}.{\tt next\_rank}$;
      \State \hspace*{2mm}\}
      \State \hspace*{2mm}${\tt all\_partial}$ $\leftarrow$ $\cup_{i=0}^{k-1} {\tt partial_i}$;
      \State \hspace*{2mm}${\tt confirm\_bar} \leftarrow \min_{i=0}^{k-1} y_i$;
      \State
      \State \hspace*{2mm}// determine fully-confirmed blocks and SCB
      \State \hspace*{2mm}$L \leftarrow \emptyset$;
      \State \hspace*{2mm}{\bf foreach} $(B,\widehat{B})\in {\tt all\_partial}$ \{
      \State \hspace*{6mm}{\bf if} ($\widehat{B}.{\tt rank} < {\tt confirm\_bar}$)
      {\bf then} $L \leftarrow L \cup \{(B, \widehat{B})\}$;
      \State \hspace*{2mm}\}
      \State \hspace*{2mm}sort blocks in $L$ by ${\tt rank}$, tie-breaking favoring smaller chain id;
      \State \hspace*{2mm}{\bf return} $L$;
      \State \}
\end{algorithmic}
\normalsize
\end{framed}
\end{subfigure}
\vspace*{2mm}
\caption{Pseudo-code of the $(k, p, \lambda, T)$-\tool protocol.
}
\label{fig:code}
\vspace*{-3mm}
\end{figure*}


\section{Implementation Details}
\label{sec:full}

We call the Nakamoto consensus protocol {\em $(p, \lambda, T)$-Nakamoto}, where the hash of a valid block needs to have $\log_2 \frac{1}{p}$ leading zeros, $\lambda$ is the security parameter (i.e., the length of the hash output), and $T$ is the number of blocks that we remove from the end of the chain in order to obtain partially-confirmed blocks. We call the \tool protocol {\em $(k, p, \lambda, T)$-\tool}, where $\lambda$ and $T$ are the same as above, and where $k$ is the number of chains in \tool. In
{\em $(k, p, \lambda, T)$-\tool}, the hash of a valid block should have $\log_2 \frac{1}{kp}$ leading zeros.

For simplicity, the value of $k$ is fixed in our current design
of \tool---adjusting $k$ on-the-fly could potentially be possible
via a view change mechanism,
but we consider it as future work.
The value of $T$ in \tool can be readily adjusted on-the-fly by individual nodes, without needing any coordination.
Because the security of \tool inherits from the security of Nakamoto consensus, using  different $T$ values in \tool has a similar effect as in Nakamoto consensus. For example, a user can use a larger $T$ for a higher security level. To simplify discussion, however, we consider some fixed $T$ value.

\Paragraph{Overview.}
Figure~\ref{fig:code} gives the pseudo-code of \tool, as run by each node.
In the main loop (Line~\ref{code:mainloop1} to \ref{code:mainloop2}) of \tool, a node receives messages from others, makes one attempt for solving the PoW (i.e., makes one query to the random oracle or hash function), and then sends out messages. The messages contain \tool blocks. Such a main loop is exactly the same as in Nakamoto consensus~\cite{kiffer18,pass2017analysis}.
Note that the block generation rate is significantly lower than the rate of queries to the random oracle. Hence, most of the time, the main loop will not have any new messages to receive/send. In an actual implementation, such sending/receiving of messages can be done in a separate thread.

The function OutputSCB() can be invoked whenever needed. It produces the current SCB,
by exactly following the description in Section~\ref{sec:ourorder}.

\Paragraph{Key data structures.}
Each \tool node maintains sets $V_0$ through $V_{k-1}$. $V_i$ initially contains the genesis block for chain $i$. During the execution, $V_i$ is a tree of blocks, containing all those blocks with a path to that genesis block. We use $W_i$ to denote the longest path (from the genesis block to some leaf) on this tree.
All blocks on $W_i$, except the last $T$ blocks, are {\em partially-confirmed}.

The Merkle tree $M$ is constructed by using the hash of the very last block on $W_i$ ($0\le i\le k-1$) as the $k$ leaves. Each node further maintains a ${\tt trailing}$ variable, which is the hash of the {\em trailing block}. The {\em trailing block} is the block with the largest ${\tt next\_rank}$ value, among all blocks in $\cup_{i=0}^{k-1}V_i$. (Note that the trailing block may or may not be in $\cup_{i=0}^{k-1}W_i$.) If there are
multiple such blocks, we let the trailing block be the one with the smallest chain id. Both $M$ and ${\tt trailing}$ should be properly updated whenever the node receives new blocks.


\Paragraph{Block generation/verification.}
A block $B$ in \tool consists of some transactions, a fresh nonce, the Merkle root, and a $B.{\tt trailing}$ field. All these fields are fed into the hash function, during mining. Section~\ref{sec:individualchain}
explains why we include the Merkle root. The following explains the $B.{\tt trailing}$ field.

%
%

As explained in Section~\ref{sec:ourorder}, regardless of which chain the new block $B$ ends up belonging to, we want $B$ to help that chain to increase its rank to
catch up with the ``longest chain'' (in terms of rank). To do so, all we need is to set $B$'s ${\tt next\_rank}$ to be large enough. A naive design is to let the creator of $B$ directly set $B$'s ${\tt next\_rank}$, based on its local $V_i$'s. Such a design enables the adversary to pick the maximum possible value\footnote{${\tt next\_rank}$ must have a finite domain in actual implementation.} for that field. Doing so exhausts the possible values for ${\tt next\_rank}$, since the ${\tt next\_rank}$ of blocks on a given chain needs to keep increasing.

This is why in \tool, each node maintains a ${\tt trailing}$ variable (i.e., the hash of the trailing block\footnote{Our trailing block is the block with the largest ${\tt next\_rank}$ in $\cup_{i=0}^{k-1}V_i$.
One could further require the trailing block to be in $\cup_{i=0}^{k-1}W_i$.
The security guarantees of \tool (in Section~\ref{sec:analysis}) and our proofs also hold for such an alternative design.}). The miner sets $B.{\tt trailing}\leftarrow {\tt trailing}$, and $B.{\tt trailing}$ (as part of $B$) is fed into the hash function. When a node $u$ receives a new block $B$, the node will verify whether it has already seen some block $C$ whose hash equals $B.{\tt trailing}$. If so, $u$ sets ${\tt next\_rank}$ of $B$ to be the same as the ${\tt next\_rank}$ of $C$. Otherwise $B$ will not be accepted. Doing so prevents the earlier attack.

Note that the adversary can still lie, and claim some arbitrary block
as its trailing block. Theorem~\ref{the:main} in
Section~\ref{sec:analysis}, however, will prove that this does {\em
not} cause any problem.  Intuitively, by doing so, the adversary is
simply refusing to help a chain to grow its rank to catch up.  But the
next honest block on the chain will enable the chain to grow its rank
properly, and to immediately catch up.


\Paragraph{Block attachment.}
Some information about a block is not available until after the block is successfully mined, and we include such information in an {\em attachment} for the block. The attachment for a block is always stored/disseminated together with the block.
We always use $\widehat{B}$ to denote an attachment for a block $B$. (Note that $\widehat{B}$ is not fed into the hash function when we compute the hash of $B$.)  Specifically, let $i$ be the last $\log_2 k$ bits in $B$'s hash. $\widehat{B}$ will include leaf $i$ of the Merkle tree, and the corresponding Merkle proof.
For convenience, $\widehat{B}$ also contains the hash of $B$, in its $\widehat{B}.{\tt hash}$ field.

An attachment $\widehat{B}$ further contains a ${\tt rank}$ value and a ${\tt next\_rank}$ value, for the block $B$. (The block $B$ itself actually does not have a ${\tt rank}$ or ${\tt next\_rank}$ field.) When a node receives a new block $B$ and its attachment $\widehat{B}$, the node independently computes the proper values for $\widehat{B}.{\tt rank}$ and $\widehat{B}.{\tt next\_rank}$ based on its local information, and use those values (Line~\ref{code:rankstart} to \ref{code:rankend} in Figure~\ref{fig:code}).
Lemma~\ref{lemma:merkle} will prove that except with an exponentially small probability (i.e., excluding hash collisions and so on), for any block $B$, all honest nodes will assign {\em exactly the same value} to $\widehat{B}.{\tt rank}$ ($\widehat{B}.{\tt next\_rank}$).
Finally, a genesis block $B$ always has $\widehat{B}.{\tt rank}=0$ and $\widehat{B}.{\tt next\_rank}=1$.

\section{Security Guarantees of \tool}
\label{sec:analysis}

Our analysis results presented in this section hold under {\em all
  possible strategies} of the adversary.

\subsection{Overview of Guarantees}

\Paragraph{Formal framework.}
Our formal framework directly follows several prior works (e.g., \cite{kiffer18,pass2017analysis}). All executions we consider are of polynomial length with respect to the security parameter $\lambda$.
We model hash functions as random oracles.
The execution of the system comprises a sequence of {\em ticks}, where a {\em tick} is the amount of time needed to do a single proof-of-work query to the random oracle by an honest node. Hence in each tick, each honest node does one such query, while the adversary does up to $fn$ such queries. We allow these $fn$ queries to be done sequentially, which only makes our results stronger. Define $\Delta$ (e.g., $2\times 10^{12}$) to be $\delta$ (e.g., $2$ seconds) divided by the duration of a tick (e.g., $10^{-12}$ second) --- hence a message sent by an honest node will be received by all other honest nodes within $\Delta$ ticks.
A block is an {\em honest block} if it is generated by some honest node, otherwise it is a {\em malicious block}. For two sequences $S_1$ and $S_2$, $S_2$ is a {\em prefix} of $S_1$ iff $S_1$ is the concatenation of $S_2$ and some sequence $S_3$. Here $S_3$ may be empty --- hence $S_1$ is also a prefix of itself.
Finally, recall from Section~\ref{sec:full} the definitions of
$(p, \lambda, T)$-Nakamoto
and $(k, p, \lambda, T)$-\tool.

\Paragraph{Main theorem on \tool.}
We will eventually prove the following:
\begin{theorem}
\label{the:main}
Consider any given constant $f<\frac{1}{2}$. Then there exists some positive constant $c$ such that for all $p\le \frac{1}{c \Delta n}$ and all $k\ge 1$, the $(k, p, \lambda, T)$-\tool protocol satisfies all the following properties, with probability at least $1- k\cdot exp(-\Omega(\lambda)) - k\cdot exp(-\Omega(T))$:
\begin{itemize}
\item {\bf (growth)} On any honest node, the length of each of the $k$ chains increases by at least $T$ blocks every $\frac{2T}{pn}$ ticks.
\item {\bf (quality)} On any honest node and at any time, every $T$ consecutive blocks on any of the $k$ chains contain at least $\frac{1-2f}{1-f}T$ honest blocks.
\item {\bf (consistency)} Consider the SCB $S_1$ on any node $u_1$ at any time $t_1$, and the SCB $S_2$ on any node $u_2$ at any time $t_2$.\footnote{Here $u_1$ ($t_1$) may or may not equal $u_2$ ($t_2$).} Then either $S_1$ is a prefix of $S_2$ or $S_2$ is a prefix of $S_1$. Furthermore, if ($u_1=u_2$ and $t_1 < t_2$) or
    ($u_1\ne u_2$ and $t_1 +\Delta < t_2$), then $S_1$ is a prefix of $S_2$.
\item {\bf (quality-growth)}
For all integer $\gamma\ge 1$, the following property holds after the very first $\frac{2T}{pn}$ ticks of the execution: On any honest node, in every
    $(\gamma +2)\cdot \frac{2T}{pn}+2\Delta$ ticks, at least
    $\gamma\cdot k \cdot  \frac{1-2f}{1-f}T$ honest blocks are newly added to SCB.
\end{itemize}
\end{theorem}

\Paragraph{Values of $c$, $\lambda$, and $T$.}
The value of $c$ in Theorem~\ref{the:main} will be {\em exactly the same} as the $c$ in Theorem~\ref{the:blackbox} next.
If we want the properties in Theorem~\ref{the:main} to hold with probability $1-\epsilon$, then both $\lambda$ and $T$ should be $\Theta(\log\frac{1}{\epsilon} + \log k)$. The value of $\epsilon$ needed by a real application (e.g., a cryptocurrency system) is typically orders of magnitude smaller than $\frac{1}{k}$ --- hence $\lambda$ and $T$ are usually just $\Theta(\log\frac{1}{\epsilon})$.

\Paragraph{Four properties.}
The {\bf growth} and {\bf quality} in Theorem~\ref{the:main} are about the individual component chains in \tool. For {\bf consistency}, considering individual chains obviously is not sufficient. Hence, Theorem~\ref{the:main} proves {\bf consistency}\footnote{Different prior works~\cite{garay2015bitcoin,pass2017analysis} define consistency slightly differently. Our definition here is either equivalent or stronger than those in the prior works.} for the SCB, which is the final total order of the fully-confirmed blocks.
Theorem~\ref{the:main} also proves the {\bf quality-growth} of SCB, showing that SCB will incorporate more honest blocks at a certain rate. Ultimately,
{\bf consistency} corresponds to the {\em safety} of \tool, while {\bf quality-growth} captures the {\em liveness} of \tool.

\subsection{Existing Result as a Building Block}

Our proof later will invoke the following theorem from \cite{pass2017analysis} on Nakamoto consensus.
\begin{theorem}
\label{the:blackbox} {\em (Adapted from Corollary 3 in \cite{pass2017analysis}.)}
Consider any given constant $f<\frac{1}{2}$. Then there exists some positive constant $c$ such that for all $p\le \frac{1}{c \Delta n}$, the $(p, \lambda, T)$-Nakamoto protocol satisfies all the following properties, with probability at least $1- exp(-\Omega(\lambda)) - exp(-\Omega(T))$:
\begin{itemize}
\item {\bf (growth)} On any honest node, the length of the chain increases by at least $T$ blocks every $\frac{2T}{pn}$ ticks.
\item {\bf (quality)} On any honest node and at any time, every $T$ consecutive blocks on the chain contain at least $\frac{1-2f}{1-f}T$ honest blocks.
\item {\bf (consistency)} Let $S_1$ ($S_2$) be the sequence of blocks on the chain   on any node $u_1$ ($u_2$) at any time $t_1$ ($t_2$), excluding the last $T$ blocks on the chain.
    Then either $S_1$ is a prefix of $S_2$ or $S_2$ is a prefix of $S_1$.
\end{itemize}
\end{theorem}
\noindent
Combing the {\bf growth} and {\bf quality} properties
immediately leads to {\bf quality-growth} for $(p, \lambda, T)$-Nakamoto:
\begin{itemize}
\item {\bf (quality-growth)} Let $S$ be the sequence of blocks on the chain on any given honest node, excluding the last $T$ blocks. Then after the very first $\frac{2T}{pn}$ ticks of the execution, in every $\frac{2T}{pn}$ ticks, at least $\frac{1-2f}{1-f}T$ honest blocks are newly added to $S$.
\end{itemize}


\subsection{Proof for Theorem~\ref{the:main}}

\Paragraph{Overview of proof for Theorem~\ref{the:main}.}
Our first step is to obtain a {\em reduction} from $(p, \lambda, T)$-Nakamoto to $(k, p, \lambda, T)$-\tool. Consider any given adversary $\mathcal{A}$ for $(k, p, \lambda, T)$-\tool, and any given chain $i$ ($0\le i\le k-1$) in the execution of $(k, p, \lambda, T)$-\tool against $\mathcal{A}$. Our reduction step will show that there exists some adversary $\mathcal{A}'$ for $(p, \lambda, T)$-Nakamoto with the following property: Except some exponentially small probability and after proper mapping from blocks in $(k, p, \lambda, T)$-\tool to blocks in $(p, \lambda, T)$-Nakamoto, the behavior of chain $i$ in the execution of $(k, p, \lambda, T)$-\tool against $\mathcal{A}$ follows {\em exactly the same distribution} as the behavior of the (single) chain in the execution of $(p, \lambda, T)$-Nakamoto against $\mathcal{A}'$. Such a reduction implies that existing properties on $(p, \lambda, T)$-Nakamoto  directly carry over to each individual chain in $(k, p, \lambda, T)$-\tool.

Our second step is to show that conditioned upon all the existing properties (in Theorem~\ref{the:blackbox}) on $(p, \lambda, T)$-Nakamoto holding for each individual chain in $(k, p, \lambda, T)$-\tool, the SCB generated in \tool must satisfy the properties in Theorem~\ref{the:main}, except with some exponentially small probability.
The reasoning in this step will center around the ${\tt rank}$ and ${\tt next\_rank}$ values of the blocks.


\Paragraph{Formal concepts needed for reduction.}
Consider any (black-box and potentially randomized) adversary $\mathcal{A}$ for \tool. Define $\tt{EXEC}(\mbox{\tool}, k,p, \lambda, T, \mathcal{A})$ to be the random variable denoting the joint states of all the honest nodes and the adversary throughout (i.e., at every tick) the entire execution resulted from running $(k,p, \lambda, T)$-\tool against $\mathcal{A}$. Define random variable ${\tt Chainview}_i(\tt{EXEC}(\mbox{\tool}, k,p, \lambda, T, \mathcal{A}))$ to be the joint state of chain $i$ on every honest node throughout this execution. Recall that in \tool, a node maintains $k$ sets of blocks, $V_0$ through $V_{k-1}$, and chain $i$ corresponds to the longest path in $V_i$.
One could imagine that for each $(u,t)$ pair, ${\tt Chainview}_i()$ contains a component describing chain $i$ on the honest node $u$ at tick $t$.

We similarly define $\tt{EXEC}(\mbox{Nakamoto}, p, \lambda, T, \mathcal{A}')$ for the execution resulted from running $(p, \lambda, T)$-Nakamoto against adversary $\mathcal{A}'$. Also, we similarly define random variable  ${\tt Chainview}(\tt{EXEC}(\mbox{Nakamoto}, p, \lambda, T, \mathcal{A}'))$ to be the joint
state of the (single) chain
on every  honest node throughout this execution.

For two random variables $X$ and $Y$ over some finite domain $Z$, define their {\em variation distance} to be $||X-Y|| = 0.5\sum_{z\in Z}|\Pr[X=z] - \Pr[Y=z]|$. If $X$ and $Y$ are both parameterized by $\lambda$, we say that $X(\lambda)$ is {\em strongly statistically close}
to $Y(\lambda)$ iff $||X-Y|| = exp(-\Omega(\lambda))$.

\Paragraph{Our reduction lemma.}
The following is our reduction lemma:
\begin{lemma}
\label{lemma:reduction}
Consider any given $i$ where $0\le i\le k-1$, and any given adversary $\mathcal{A}$ for $(k,p, \lambda, T)$-\tool. There exists some adversary $\mathcal{A}'_i$ for $(p, \lambda, T)$-Nakamoto such that
the following two random variables are strongly statistically close:
\begin{itemize}
\item ${\tt Chainview}_i(\tt{EXEC}(\mbox{\tool}, k,p, \lambda, T, \mathcal{A}))$
\item $\sigma_i^\tau({\tt Chainview}(\tt{EXEC}(\mbox{Nakamoto}, p, \lambda, T, \mathcal{A}'_i))$
\end{itemize}
Here $\tau$ is the randomness in $\tt{EXEC}(\mbox{Nakamoto}, p, \lambda, T, \mathcal{A}'_i)$. The random variable $\sigma_i^\tau({\tt Chainview}())$ is the same as ${\tt Chainview}()$, except that we replace each block $B'$ in ${\tt Chainview}()$ with another block $\sigma_i^\tau(B')$. The mapping $\sigma_i^\tau()$ is some one-to-one mapping from each block $B'$ among all the blocks
on all the nodes in $(p, \lambda, T)$-Nakamoto
to some block $\sigma_i^\tau(B')$ on all the nodes in $(k,p, \lambda, T)$-\tool. The mapping $\sigma_i^\tau()$ guarantees that i) $B'$ is an honest block iff $\sigma_i^\tau(B')$ is an honest block, and ii) $B'$ extends from $A'$ iff $\sigma_i^\tau(B')$ extends from  $\sigma_i^\tau(A')$.
\end{lemma}

The crux of proving this lemma is to construct the adversary $\mathcal{A}'_i$. In our proof, $\mathcal{A}'_i$ simulates a certain execution of \tool against $\mathcal{A}$. More precisely, $\mathcal{A}'_i$ simulates all the \tool nodes, as well as the adversary $\mathcal{A}$ in a black-box fashion. The adversary $\mathcal{A}'_i$ simultaneously interacts with the (real) nodes running $(p, \lambda, T)$-Nakamoto, while ensuring that the simulated \tool execution and the real Nakamoto execution are properly ``coupled''. We defer the complete proof of this lemma to Appendix~\ref{app:reductionformal} through \ref{app:reductionproof}.
To be fully rigorous,
the complete proof needs additional formalism and also needs to fully specify the $(p, \lambda, T)$-Nakamoto protocol by pseudo-code.


\Paragraph{From individual chains to SCB.}
The following lemma (proof in Appendix~\ref{app:carryover}) establishes the connection from the individual chains in \tool to the SCB in \tool:
\begin{lemma}
\label{lemma:carryover}
If the three properties in Theorem~\ref{the:blackbox} hold for each of the $k$ chains in $(k,p, \lambda, T)$-\tool, then with probability at least $1-exp(-\Omega(\lambda))$, the SCB in \tool satisfies the {\bf consistency} and {\bf quality-growth} properties in Theorem~\ref{the:main}.
\end{lemma}

\Paragraph{Final proof.}
Using all the lemmas, we can now prove Theorem~\ref{the:main}:
\begin{proof} (for Theorem~\ref{the:main})
We set the constant $c$ in Theorem~\ref{the:main} to be the same as the $c$ in Theorem~\ref{the:blackbox}. For any given $i$ where $0\le i\le k-1$, Lemma~\ref{lemma:reduction} and Theorem~\ref{the:blackbox} tell us that for chain $i$  in $(k,p, \lambda, T)$-\tool, with probability at least $1- exp(-\Omega(\lambda)) - exp(-\Omega(T)) - exp(-\Omega(\lambda))$, the three properties in Theorem~\ref{the:blackbox} hold for that chain. Hence with probability at least $1- k\cdot exp(-\Omega(\lambda)) - k\cdot exp(-\Omega(T))$, the properties in Theorem~\ref{the:blackbox} hold for {\em all} $k$ chains in $(k,p, \lambda, T)$-\tool.
The {\bf growth} and {\bf quality} properties in Theorem~\ref{the:main} then directly follow. Applying Lemma~\ref{lemma:carryover} further leads to the {\bf consistency} and {\bf quality-growth} properties in Theorem~\ref{the:main}.
\end{proof}


\subsection{Discussion and Comparison}



\Paragraph{Plug-in alternative results.}
Our analysis invokes the results in \cite{pass2017analysis}.
The analysis in \cite{pass2017analysis} is just one of the many works~\cite{garay2015bitcoin,pow-difficulty,kiayias15tradeoff,kiffer18,pass2017analysis} that analyze Nakamoto-style protocols.
A highlight of our proof on \tool is that it invokes the existing guarantees on $(p, \lambda, T)$-Nakamoto as {\em black-box}. Hence alternative results on $(p, \lambda, T)$-Nakamoto directly translates to alternative results on \tool.

Specifically, Theorem~\ref{the:blackbox}
(adopted from \cite{pass2017analysis}) has the following three {\em quantitative measures}:
\begin{itemize}
\item The value of $c$, where $\frac{1}{c\Delta n}$ is the upper limit on $p$.
\item The value of $x = \frac{2T}{pn}$ ticks (i.e., the {\em growth rate}), which is the time needed for the chain length to grow by $T$ blocks.
\item The value of $y = \frac{1-2f}{1-f}T$ (i.e., the {\em quality rate}), which is the number of honest blocks among every $T$ consecutive blocks on the chain.
\end{itemize}
Other analyses~\cite{garay2015bitcoin,pow-difficulty,kiayias15tradeoff,kiffer18} have obtained alternative results on $x$ and $y$ values for $(p, \lambda, T)$-Nakamoto. They have also derived various sufficient conditions\footnote{For example, the sufficient condition derived in \cite{pass2017analysis} is that \linebreak $\alpha(1-2(\Delta+1)\alpha)\ge (1+\eta)\beta$ for some positive constant $\eta$.} for consistency in $(p, \lambda, T)$-Nakamoto, which all ultimately translate to requirement on the value of $c$.

We can directly plug in alternative $c$, $x$, and $y$ values from alternative analyses on $(p, \lambda, T)$-Nakamoto to obtain alternative results on \tool. If we do so, then the value of $c$ in Theorem~\ref{the:main} will be {\em exactly the same} as the given $c$. The {\bf consistency} property in Theorem~\ref{the:main} will remain unchanged. The remaining properties in Theorem~\ref{the:main} will become:
\begin{itemize}
\item {\bf (growth)} The length of each of the $k$ chains on each honest node will increase by at least $T$ blocks every $x$ ticks.
\item {\bf (quality)} On any honest node and at any time, every $T$ consecutive blocks on any of the $k$ chains must contain at least $y$ honest blocks.
\item {\bf (quality-growth)}
For all integer $\gamma\ge 1$, after the very first $x$ ticks of the execution, on any honest node in every
    $(\gamma +2)\cdot x+2\Delta$ ticks, at least
    $\gamma\cdot k \cdot y$ honest blocks are newly added to SCB.
\end{itemize}

%

\Paragraph{Comparing with prior results.}
With the above discussion, we can now easily compare Theorem~\ref{the:main} with any of the prior results~\cite{garay2015bitcoin,pow-difficulty,kiayias15tradeoff,kiffer18,pass2017analysis}.
Consider the result from any such prior analyses, with certain resulting $c$, $x$, and $y$ values. (This also implies that
in that particular result, the rate of {\bf quality-growth} is
$y$ new honest blocks every $x$ ticks.)
Now let us plug in such $c$, $x$, and $y$ values into Theorem~\ref{the:main}. Then \tool will provide exactly the same quantitative guarantees as that prior result, in terms of the value of $c$, the growth rate, and the quality rate. The only difference will be regarding the {\bf quality-growth} of SCB. Under the given $x$ and $y$, for all integer $\gamma\ge 1$, in \tool $\gamma\cdot k\cdot y$ honest blocks are newly added to the SCB every
$(\gamma +2)\cdot x+2\Delta$ ticks. Since $2\Delta$ is dominated by $(\gamma +2)\cdot x$, the average rate of honest blocks being added to the SCB is about $k\cdot y$ honest blocks every $x$ ticks, in the long term. Such a rate is $k$ times of the rate of {\bf quality-growth} in the corresponding prior result. This also intuitively explains why \tool increases throughput by $k$ times.\footnote{Of course, $k$ cannot be unbounded. In practice, increasing $k$ beyond a certain point will cause a non-trivial increase in $\Delta$, when the system starts to saturate the network bandwidth. Our experiments later will quantify this.}

\Paragraph{Confirmation latency.}
Finally, Theorem~\ref{the:main} also indirectly gives \tool's guarantee on transaction {\em confirmation latency} (also called {\em wait-time} in \cite{garay2015bitcoin,pass2017analysis}): Assume that we want the properties in Theorem~\ref{the:main} to hold with $1-\epsilon$ probability, and let us invoke the theorem with $p = \Theta(\frac{1}{c\Delta n})$. By {\bf quality-growth}, at least one honest block is newly added to SCB within $\Theta(T\Delta) = \Theta((\log\frac{1}{\epsilon} + \log k)\Delta)$ ticks. Now if a transaction is injected into all honest nodes continuously for $\Theta((\log\frac{1}{\epsilon} + \log k)\Delta)$ ticks, this transaction must be included in the SCB after those ticks. Hence, the confirmation latency is simply $\Theta((\log\frac{1}{\epsilon} + \log k)\Delta)$ ticks.
As explained earlier, this usually becomes $\Theta(\Delta\log\frac{1}{\epsilon})$ in practical settings. Such a confirmation latency is the same as in \cite{garay2015bitcoin,pass2017analysis}, and is fundamental in all Nakamoto-style protocols: Each new block takes $\Theta(\Delta)$ ticks, and $\Theta(\log\frac{1}{\epsilon})$ new blocks are needed for the ``confidence'' to reach $1-\epsilon$.

\section{Experimental Evaluation}
\label{sec:eval}

\Paragraph{Methodology.}
We have implemented a prototype of \tool in C++, with total around 4,700
lines of code. We use Amazon's EC2
virtual machines (or EC2 instances) for evaluation. We rent \texttt{m4.2xlarge}
instances in 14 cities around the globe\footnote{This is the maximal
number of cities with such instances.}, with each instance having
$8$ cores and $1$Gbps bandwidth.
The one-way latency between two random EC2 instances is about
$90$-$140$ms, which is consistent with AlgoRand's experiments~\cite{algorand} and with measurements in Bitcoin and Ethereum~\cite{gencer-fc18}. To avoid excessive monetary expenses on EC2, we run two sets of experiments. Our {\em macro experiments} run $12,000$ to $50,000$ \codename nodes on up to $1,000$ EC2 instances to evaluate the end performance of \tool, while our {\em micro experiments} run $1,000$ \codename nodes on $20$ EC2 instances to determine \codename internal parameters.



In all experiments, the \tool nodes form a P2P overlay by each node connecting to $8$ randomly selected peers, and per-node bandwidth is up to $20$Mbps, since we run $50$ nodes per EC2 instance. This setup is the same as in the experiments of AlgoRand~\cite{algorand} and Conflux~\cite{conflux}.
All results reported are averaged over $5$
runs, each lasting until measurements stabilize in that run.



\subsection{Choosing Block Size and Block Interval in \tool}
\label{sec:micro}


We first use micro experiments to measure the {\em block propagation delay}
({\em BPD}) for a single block (with
no parallel propagations) to reach $99\%$ of the nodes. We consider different bandwidth configurations where the per-node bandwidth ranges from $8$Mbps to $20$Mbps. We observe that the BPD for \blocksize blocks is about $1.7$-$1.9$ seconds, across all our bandwidth configurations. Similar BPD values are observed for block sizes ranging from $10$ KB to $64$ KB.




We further observe that such BPD does not significantly increase as the network size
increases: Even in our macro experiments with $50,000$ nodes,
the BPD values are still only about $3.2$ seconds for $10$-$20$ KB blocks.
This is consistent with theoretical expectations about
random graphs,\footnote{In random Erdos-Renyi graphs, roughly speaking, the average number of hops between two nodes increases only logarithmically with the number of nodes.} and the fact that BPD is proportional to the average number
of hops between two nodes.


Smaller blocks enable better decentralization, but also incur more protocol overheads. Taking all factors into account, we choose a block size of \blocksize for \tool. We then set $p$ to correspond to a block interval of about $10$ seconds on {\em each} chain. Based on Section~\ref{sec:analysis} and results in  \cite{pass2017analysis}, a block interval of $5\times$ BPD is sufficient to tolerate $f =0.43$.

\subsection{Efficient Parallel Propagation of Blocks}


At the network level, a key difference between \tool and other
high-throughput protocols (e.g.,  Algorand~\cite{algorand}) is that
with its large number of parallel chains, \tool propagates a large
number of small blocks in parallel. In contrast, protocols such as Algorand
usually propagate a single large block (e.g., of $1$ MB).

Now for \tool, the critical empirical question is whether propagating many parallel blocks will have significant negative impact on BPD, as compared to propagating a single such block. We use micro experiments to answer this critical question.
Figure~\ref{fig:20KB}(left) plots the BPD for \blocksize blocks (results for block sizes ranging from $10$-$64$KB are similar and not shown), as a function of number of parallel block propagations. For example, ``40 parallel blocks/sec'' means that we inject $40$ new blocks (of size \blocksize each) into the network per second. The figure shows that under $20$Mbps raw bandwidth, even $60$ parallel blocks per second will not
cause any substantial increase in BPD.

Figure~\ref{fig:20KB}(right) presents the same results from a different perspective --- it plots how BPD changes as the fraction of raw bandwidth used by the parallel block propagation increases. It shows that, consistently under all our bandwidth configurations, parallel block propagation can effectively utilize a rather significant fraction (about $50$\%) of the raw bandwidth, without significant negative impact on BPD. This simple yet important finding lays the empirical
foundation for parallel chain designs. In particular, we can hope
\tool to eventually achieve a throughput approaching a significant fraction of the raw network bandwidth, by using a sufficient number of parallel chains.

\begin{figure}[t]
\pgfplotsset{compat=newest,
legend style={font=\footnotesize},
label style={font=\footnotesize},
tick label style={font=\footnotesize},
title style={font=\footnotesize}}

\pgfplotsset{footnotesize,height=4cm,width=0.25\textwidth}

\begin{center}
\begin{tikzpicture}[scale=1]
\begin{axis}[
legend columns=3,
legend entries={8Mbps;,16Mbps;,20Mbps},
legend to name=fig:num_par,
xmin=0,
ymin=0,
ymax=8.2,
xlabel={\# parallel blocks/second\hspace*{2mm}},
ylabel={BPD (seconds)},
grid=major]
\addplot[
    color=blue,
    mark=square,
    ]
    coordinates {
      (0,1.7)(12.5,2.141)(16.5,2.33)(25.0,2.311)(33.0,3.838)(37.5,5.975)
    };
\addplot[
    color=green,
    mark=*,
    ]
    coordinates {
      (0,1.8)(25.0,1.961)(33.0,1.955)(50.0,2.05)(66.0,3.453)(75.0,5.559)
    };

\addplot[
    color=red,
    mark=x,
    ]
    coordinates {
      (0,1.9)(31.25,2.036)(41.25,2.097)(62.5,2.165)(82.5,3.826)(93.75,5.86)
    };
\end{axis}
\end{tikzpicture}
\begin{tikzpicture}[scale=1]
\begin{axis}[
ymin=0,
ymax=8.2,
xlabel={bandwidth utilization (\%) },
ylabel={BPD (seconds)},
grid=major]
\addplot[
    color=blue,
    mark=square,
    ]
    coordinates {
    (25,2.141)(33,2.33)(50,2.311)(66,3.838)(75,5.975)
    };
\addplot[
    color=green,
    mark=*,
    ]
    coordinates {
    (25.0,1.961)(33.0,1.955)(50.0,2.05)(66.0,3.453)(75.0,5.559)
    };

\addplot[
    color=red,
    mark=x,
    ]
    coordinates {
    (25,2.036)(33,2.097)(50,2.165)(66,3.826)(75,5.86)
    };
\end{axis}
\end{tikzpicture}
\\
\ref{fig:num_par}
\end{center}
\vspace*{-0mm}
\caption{
BDP under different per-node bandwidth configurations ($8$-$20$Mbps).
[\textbf{left}]: BPD vs. number of blocks
propagated in parallel per second.  The BPD for a
  single block (non-parallel case) is plotted as x-axis value of $0$.
  [\textbf{right}]: BPD vs. fraction of raw bandwidth utilized
  by parallel block propagation.}
  \vspace*{-2mm}
\label{fig:20KB}
\end{figure}
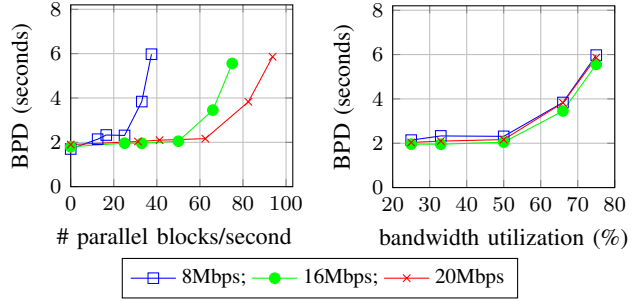


\begin{figure*}[htbp]
\begin{subfigure}{.33\textwidth}
\pgfplotsset{compat=newest,
legend style={font=\footnotesize},
label style={font=\normalsize},
tick label style={font=\footnotesize},
title style={font=\footnotesize}}
\centering\captionsetup{width=.9\linewidth}
\begin{tikzpicture}[font=\footnotesize]
\begin{axis}[
scale=0.5,
height=1.2*\axisdefaultheight,
width=\textwidth*1.7,
xtick={8,12,16,20},
ylabel=\footnotesize{txns throughput (Mbps)},
xlabel=\footnotesize{available bandwidth (Mbps)},
ybar=1pt,
tick align=inside,
ymin=0,
ymax=12,
xmin = 6,
xmax=22,
xlabel shift = -5 pt, 
ylabel shift = -5 pt, 
tick pos=left,
bar width=14pt,
nodes near coords,
every node near coord/.append style={font=\scriptsize},
]
\addplot
coordinates {(8,3.84) (12, 5.84) (16, 7.76) (20, 9.68)};
\end{axis}

\begin{axis}[
  scale=0.5,
  height=1.2*\axisdefaultheight,
  width=\textwidth*1.7,
  axis y line*=right,
  axis x line=none,
  ymin=0,
  ymax=3000,
  ylabel shift = -5 pt, 
  ylabel=\footnotesize{txns per second (tps)},
]
\addplot [only marks, mark=-]
coordinates {(8,960) (12, 1460) (16, 1940) (20, 2420)};

\end{axis}
\end{tikzpicture}
\vspace*{-6mm}
\caption{Transaction throughput in Mbps (left y-axis) and transactions/second (right y-axis).}
\label{fig:throughput}
\end{subfigure}
\begin{subfigure}{.33\textwidth}
\pgfplotsset{compat=newest,
legend style={font=\footnotesize},
label style={font=\normalsize},
tick label style={font=\footnotesize},
title style={font=\footnotesize}}
\centering\captionsetup{width=.9\linewidth}
\begin{tikzpicture}[trim axis right,font=\footnotesize]
\begin{axis}[
scale=0.5,
height=1.2*\axisdefaultheight,
width=\textwidth*1.7,
legend columns=2,
legend entries={exp. observed;, theoretically expected},
legend to name=fig:dec,
xtick={8,12,16,20},
ylabel=\footnotesize{decentr. factor (blocks/sec)},
ylabel shift = -5 pt, 
xlabel=\footnotesize{available bandwidth (Mbps)},
xlabel shift = -5 pt, 
ybar=1pt,
xtick align=inside, 
bar width=9pt,
xmin = 6,
xmax = 22,
ymin=0,
ymax=75,
tick pos=left,
nodes near coords,
every node near coord/.append style={font=\scriptsize},
]
\addplot [pattern=north east lines,]
coordinates {(8,24.7) (12, 37.0) (16, 49.6) (20, 61.8)};
\addplot
coordinates {(8,25) (12, 37) (16, 50) (20, 62)};
\end{axis}

\end{tikzpicture}
\\
\ref{fig:dec}
\caption{Decentralization factor.}
\label{fig:decentral}
\end{subfigure}
\begin{subfigure}{.33\textwidth}
\pgfplotsset{compat=newest,
legend style={font=\footnotesize},
label style={font=\normalsize},
tick label style={font=\footnotesize},
title style={font=\footnotesize}}
\centering\captionsetup{width=.9\linewidth}
\begin{tikzpicture}[trim axis right,font=\footnotesize]
\begin{axis}[
scale=0.5,
height=1.2*\axisdefaultheight,
width=\textwidth*1.8,
legend columns=2,
legend entries={partial-confirm;, full-confirm},
legend to name=fig:confirmT-internal,
ylabel=\footnotesize{confirmation latency (sec)},
ylabel shift = -5 pt, 
xlabel=\footnotesize{$T$},
xlabel shift = -5 pt, 
xmin=0,
ymin=0,
ymax=630,
nodes near coords,
every node near coord/.append style={font=\scriptsize},
grid=major
]
  \addplot[
    color=blue,
    mark=*,
  ]
  coordinates {(6,59.4) (10,100.2) (15,150.5) (20,200.3) (25,250.3) (30,299.9)};
  \addplot[
    color=black,
    mark=x,
  ]
  coordinates {(6,188.6) (10,250.1) (15,324.3) (20,390.2) (25,468.5) (30,556.3)};
\end{axis}
\end{tikzpicture}
\\
\ref{fig:confirmT-internal}
\caption{Confirmation latency.}
\label{fig:confirm}
\end{subfigure}
%
%
\vspace*{-0mm}
\caption{\tool's throughput, decentralization factor, and confirmation latency.}
\vspace*{-2mm}
\end{figure*}
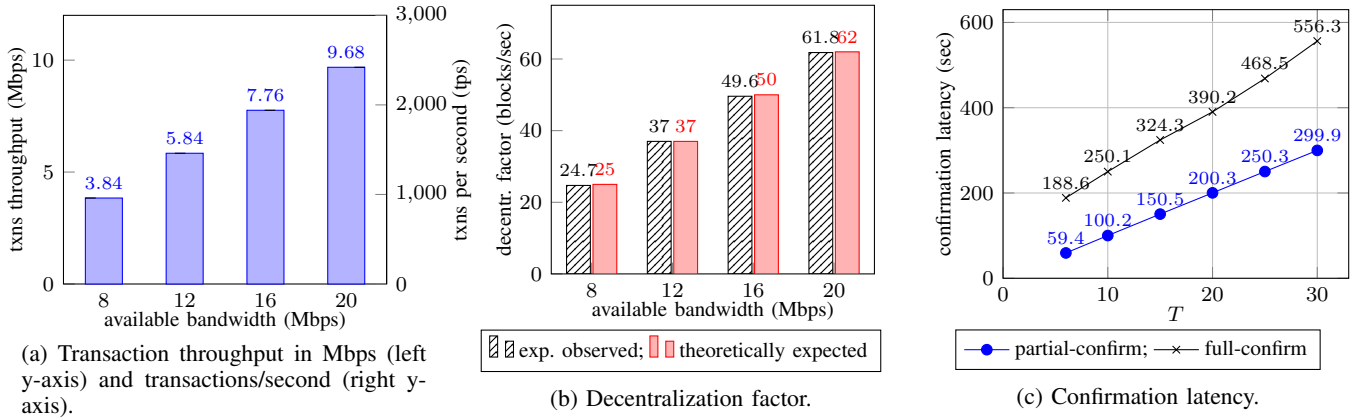

\subsection{End-to-end Performance of \tool}

We finally use macro experiments to evaluate the end-to-end performance
of \codename, in terms of its throughput, decentralization factor, and confirmation latency.
By default, all our results will be from running $12,000$ nodes on $1000$ EC2 instances, with different per-node bandwidth configurations ($8$-$20$Mbps). We have also
experimented with $50,000$ nodes on $1000$ EC2 instance, with $20$Mbps per-node bandwidth. Those results are within $1\%$ of the results under the corresponding experiment with $12,000$ node, and we do not report those separately.

We always use \blocksize blocks, with block interval being
$10$ seconds on each chain, as determined in Section~\ref{sec:micro}.
We choose $k$ according to the available
bandwidth, such that
$k\times \mbox{block\_size} / \mbox{block\_interval} \approx 0.5\times
\mbox{available\_bandwidth}$. Specifically, for $8$Mbps, $12$Mbps, $16$Mbps, and $20$Mbps per-node available bandwidth, we use $k=250$, $370$, $500$, and $620$, respectively. Since these $k$ values are not powers of $2$, we do not use the last
$\log_2 k$ bits of the block hash to decide which chain a block belongs to. Instead, let $x$ be the last $48$ bits of the block hash, and we assign the block to chain $i$ where $i = x \mod k$. Assuming the hash function is a random oracle, doing so will assign each block to a uniformly random chain, except some negligible probability.\footnote{This is not exactly uniformly random since the range of $x$ is not an exact multiple of $k$. But the impact is negligible.}
Finally, to be able to run multiple nodes on each EC2 instance, we do not have the nodes solve PoW puzzles. Instead, each node produces (``mines'') a new block after some exponentially distributed time that corresponds to the mining difficulty.




\Paragraph{Throughput.}
Figure~\ref{fig:throughput} shows that the throughput of \tool indeed
scales up roughly linearly with the available bandwidth. In fact, the throughput of \tool always reaches a rather significant fraction (about $50$\%) of the raw available network bandwidth of the system.
Under $20$Mbps available bandwidth, the throughput of
\tool is about $9.68$Mbps, or about $2,420$ transactions per second
(assuming $500$-byte average transaction size as in Bitcoin). 
As a quick comparison, \codename achieves about $550\%$
of the throughput of AlgoRand~\cite{algorand} under similar available bandwidth.
Compared to Conflux
under $20$Mbps available bandwidth~\cite{conflux}, the throughput of \codename
is about $150\%$ of the throughput of Conflux. This suggests that while explicitly focusing on simplicity, \tool still retains the high throughput property of modern blockchain designs.


\Paragraph{Decentralization factor.}
Another advantage of \codename is its {\em decentralization
factor}, in terms of the number of distinct confirmed blocks per
second. Figure~\ref{fig:decentral} shows that the decentralization
factor of \codename increases linearly with the available
bandwidth. This is expected since the number of parallel blocks (or
chains) increases with the available bandwidth. In particular,
\codename achieves a decentralization factor of about $61.8$ under $20$Mbps
available bandwidth. This is at least about $20$x higher
than those reported in experiments on previous permissionless protocols, among which Omniledger~\cite{omniledger} reported the best decentralization factor of about $3.1$ (i.e., $25$ blocks in $8.1$ seconds).




\Paragraph{Confirmation latency.}
Let $T$ be the number of blocks that we remove from the end of the
chain in order to obtain partially-confirmed blocks in each individual
chain in \tool. For example, Bitcoin and Ethereum use $T=6$ and
$T=10$ to $15$, respectively.
For comparable security, our analysis in Section~\ref{sec:analysis} suggests using a
$T$ that is $\Theta(\log k)$ larger. Given that our $k$ is no larger than $2^{14}$, we use $T=20$ to $30$ in \codename. Note that $T$ has impact only on confirmation latency, and has no impact on the throughput or decentralization factor of \tool.

Figure~\ref{fig:confirm} plots the average time for
a block to become partially-confirmed and fully-confirmed on all nodes, under $20$Mbps bandwidth configuration. It
shows that a block takes
about $1$-$5$ minutes to become partially-confirmed, and then about
another $2$-$4$ minutes to become fully-confirmed.
As a reference point, the
confirmation latencies in Bitcoin and Ethereum are presently $60$
and $3$ minutes, respectively.

For a conservative $T=30$, we have done further experiments confirming (not plotted in Figure~\ref{fig:confirm}) that the partial and full confirmation latencies remain stable under various bandwidth configurations from $8$Mbps to $20$Mbps. Putting it another
way, the latency does not deteriorate as the throughput of \tool increases.

\section{Related Work}
\label{sec:related}

\codename uses PoW under the same permissionless model as Nakamoto consensus.
There have been
works in alternative models, such as using
Proof-of-Stake~\cite{algorand,bitcoin-sok,snow-white, ouroboros} or
assuming a permissioned setting~\cite{garaysok,rscoin,honeybadger}.
PoW-based permissionless blockchain protocols have largely
followed two paradigms: The first based on extensions to Nakamoto
consensus, and the second based on utilizing classical byzantine agreement
(BA). We discuss these two categories one by one.

%

%

\Paragraph{Nakamoto consensus.}
Existing deployments of Nakamoto consensus (e.g., Bitcoin) and its variant the GHOST protocol~\cite{ghost} (e.g., Ethereum)
achieve only about $5$KByte/s throughput.
Bitcoin-NG is a variant of Nakamoto consensus, where many micro-blocks are proposed by the same proposer, after the proposer is
chosen by a key block~\cite{bitcoin-ng}. This helps to significantly improve throughput without comprising security, but at the same time, still has limited decentralization: A single block proposer is responsible for generating all (micro) blocks for an extended period of time (i.e., 10 minutes), inviting censorship attacks and DoS attacks. In comparison, \tool achieves superior decentralization --- in our experiments, in each second, up to about $60$ different proposers propose blocks concurrently.

\Paragraph{Scaling Nakamoto consensus.}
Phantom~\cite{phantom} and Conflux~\cite{conflux} have attempted to scale Nakamoto consensus, by having blocks reference more than one previous blocks. Section~\ref{sec:intro} has already discussed these two works.
  


Chainweb~\cite{martino18,quaintance18} is another attempt to scale
Nakamoto consensus, by maintaining $k$ parallel chains.
Unlike \tool, Chainweb does not have any mechanism to
prevent the adversary from focusing entirely on one of the
$k$ chains.
Chainweb's original analysis~\cite{martino18,quaintance18}
only considers a few specific attack strategies~\cite{pass2017analysis}, instead of all possible adversaries.  Fitzi et
al.~\cite{fitzi18} have shown that an adversary focusing on a specific
chain can cause the confirmation latency in Chainweb to increase {\em
quadratically} in $k$. In contrast, \tool does not suffer from such a problem.
Furthermore, different from \tool where the ${\tt rank}$ values balance all the chains,
Chainweb requires synchronous
growth of all chains, needing to periodically stall fast-growing
chains.
Kiffer et al.~\cite{kiffer18} have provided a further analysis of Chainweb,
which does not support the high-throughput claim by
Chainweb. In their analysis, a natural form of Chainweb is
  ``bounded by the same throughput as Nakamoto protocol for the same
  consistency guarantee''~\cite{kiffer18}.

Concurrent with and independent of this work, there are two online
non-refereed technical reports~\cite{bagaria18,fitzi18} proposing a
similar approach to composing multiple parallel
chains.
While their high-level designs share
similarities with ours, in their designs, one of the $k$ parallel
chains is designated as a special chain, and blocks on the other
chains are related to blocks on that special
chain~\cite{bagaria18,fitzi18}.
%
In \tool, all $k$ chains are equal and symmetric.
More importantly, these two concurrent and
independent works~\cite{bagaria18,fitzi18} do not have any
implementation details or experimental evaluation, while we have a
prototype implementation as well as large-scale evaluation on Amazon
EC2.  In particular, our large-scale experiments confirm
that propagating many parallel blocks does {\em not}
negatively impact block propagation delay, which lays the empirical
foundation for parallel chain designs.

Finally, Spectre~\cite{spectre} confirms blocks
without guaranteeing a total order, while the inclusive protocol~\cite{inclusive} includes as many non-conflicting transactions as possible. These works provide weaker consistency notions than \tool, which guarantees a total order.
These weaker notions may suffice for cryptocurrency payments,
but for smart contracts, total ordering is key in resolving state
conflicts~\cite{sequential-consistency,concurrency-db-79}
and in building towards higher abstractions of
consistency~\cite{peercensus,chainspace,omniledger}.


%



\Paragraph{Blockchain protocols relying on byzantine agreement.}
Some PoW-based permissionless blockchain protocols~\cite{peercensus,hybridconsensus,byzcoin,elastico,rapidchain,omniledger,solida} build upon classical byzantine agreement (BA) protocols. BA protocols require a committee of pre-agreed identities among
which the BA protocol can run. Such a committee can either be established using
Nakamoto consensus itself or using previous rounds of BA.


Different from \tool and other Nakamoto-style protocols, which can tolerate any constant $f< \frac{1}{2}$, BA-based blockchain protocols~\cite{algorand,solida,elastico,rapidchain,byzcoin,hybridconsensus}
all require $f < \frac{1}{3}$.
A key bottleneck in BA-based designs, as acknowledged in prior
works, is that of establishing (or replenishing) committees with
$200$-$2000$ identities each.  This large committee size is necessary
to ensure the requisite resilience $f$. The latency of replenishing
committees can be between in tens of seconds~\cite{algorand,solida},
minutes~\cite{elastico,rapidchain}, or hours~\cite{omniledger}.  After
committees are established and when there are no attacks,
BA protocols tend to have smaller confirmation
latencies~\cite{algorand,rapidchain,byzcoin} than Nakamoto-style protocols.
Finally, we point out that classical BA protocols, such as PBFT~\cite{pbft},
can be considerably more complex to implement and verify, as compared
to Nakamoto consensus.

\Paragraph{Sharding designs.}
A subset of the BA-based blockchain protocols employ sharding~\cite{elastico,omniledger,rapidchain}, where many parallel
shards process blocks in parallel. This helps to improve decentralization and throughput, at the cost of additional complexity. But due to the overheads associated with each shard, these protocols~\cite{elastico,rapidchain,omniledger} typically
use a small number of shards (no more than $25$).
The best decentralization was achieved in \cite{omniledger}, with about $25$ blocks proposed every $8$ seconds. In comparison, \tool does not involve pre-assigning nodes to its $k$ chains, and can easily use as large a $k$ as needed. The decentralization factor of \tool in our experiments (i.e., about $600$ blocks every $10$ seconds) is about $20$x higher than the sharding designs.

\section{Conclusion}
\label{sec:conclusion}

We present \codename, a permissionless proof-of-work blockchain
protocol that composes parallel instances of Nakamoto consensus
securely.  \codename has a simple implementation and a modular safety
and liveness proof. It achieves linear scaling with available
bandwidth and at least about $20$x better decentralization over prior
works.

\section*{Acknowledgments}
We thank Hung Dang, Seth Gilbert, Aquinas Hobor, Ilya Sergey, Shruti
Tople, and Muoi Tran for their helpful
feedbacks on this work. We thank Sourav Das for help on formatting the figures in this paper. This work is partly supported by the sponsors of the
Crystal Center at National University of Singapore. All opinions and
findings presented in this work are those of the authors only.

\bibliographystyle{acm}
\bibliography{paper}

\begin{appendix}


\vspace*{3mm}
\subsection{Proof for Lemma~\ref{lemma:carryover}}
\label{app:carryover}
\vspace*{3mm}

Lemma~\ref{lemma:carryover} directly follows from Lemma~\ref{lemma:mainconsistency} and Lemma~\ref{lemma:mainqg}, which we will state and prove next. To prove Lemma~\ref{lemma:mainconsistency} and Lemma~\ref{lemma:mainqg}, we will need Lemma~\ref{lemma:merkle}. Lemma~\ref{lemma:merkle} shows that except some exponentially small probability, given a block $B$ and regardless of whether $B$ is an honest block or malicious block, the value of $\widehat{B}$ (and in particular, the value of $\widehat{B}.{\tt rank}$ / $\widehat{B}.{\tt next\_rank}$) must be the same on all honest nodes. This avoids the need of reasoning about potentially different $\widehat{B}.{\tt rank}$ ($\widehat{B}.{\tt next\_rank}$) values on different honest nodes. In the following, we will state and prove Lemma~\ref{lemma:merkle} through \ref{lemma:mainqg}, one by one.

\begin{lemma}
\label{lemma:merkle}
Consider the execution of $(k, p, \lambda, T)$-\tool against any given adversary $\mathcal{A}$. With probability at least $1-exp(-\Omega(\lambda))$,
there will never be two honest nodes $u_1$ and $u_2$ adding $(B_1, \widehat{B_1})$ and $(B_2, \widehat{B_2})$ to their local set of blocks, respectively, such that $B_1=B_2$ and $\widehat{B_1}\ne \widehat{B_2}$.
\end{lemma}
\begin{proof}
Define ${\tt bad1}$ to be the event that the execution (of $(k, p, \lambda, T)$-\tool against $\mathcal{A}$) contains
some invocations $hash(x_1)$ and $hash(x_2)$ such that $x_1\ne x_2$ and $hash(x_1)= hash(x_2)$. Given that we model the hash function as a random oracle, and given that the length of the execution is $\mbox{poly}(\lambda)$, we have $\Pr[{\tt bad1}] = exp(-\Omega(\lambda))$.

Next define ${\tt bad2}$ to be the event that the execution (of $(k, p, \lambda, T)$-\tool against $\mathcal{A}$) contains some invocation of $hash()$ that belongs to one of the following two categories:
\begin{itemize}
\item For some $B$, some honest node invokes $hash(B)$ at Line~\ref{code:h1verify} of Figure~\ref{fig:code} and the verification at that line succeeds (i.e., $hash(B)$ indeed equals $\widehat{B}.{\tt hash}$), despite that $hash(B)$ has never been previous invoked (by either some honest node or the adversary) in the execution.
\item Some honest node invokes $hash()$ at Line~\ref{code:h2verify} of Figure~\ref{fig:code} and the Merkle verification at that line succeeds, despite the following: Let $hash(x_1)$, $hash(x_2)$, \ldots, $hash(x_{\log_2 k})$ be the $\log_2 k$ hash invocations done by the honest node during the successful Merkle verification. There exists some $x_i$ ($1\le i\le \log_2 k$) such that $hash(x_i)$ has never been previously invoked (by either some honest node or the adversary) in the execution.
\end{itemize}
Again given that we model the hash function as a random oracle, and given that the length of the execution is $\mbox{poly}(\lambda)$, one can easily verify that $\Pr[{\tt bad2}] = exp(-\Omega(\lambda))$.

A simple union bound then shows that with probability at least $1- exp(-\Omega(\lambda))$, neither ${\tt bad1}$ nor ${\tt bad2}$ happens. It now suffices to prove that conditioned upon neither of the bad events happening, if $B_1 = B_2$, then we must have $\widehat{B_1} = \widehat{B_2}$.
An attachment consists of five fields, namely, ${\tt hash}$, ${\tt leaf}$, ${\tt leaf\_proof}$, ${\tt rank}$, and ${\tt next\_rank}$. We will reason about these one by one.

Since $B_1 = B_2$, and since $(B_1, \widehat{B_1})$ and $(B_2, \widehat{B_2})$
have been accepted by $u_1$ and $u_2$ respectively,
we trivially have
$\widehat{B_1}.{\tt hash} = \widehat{B_2}.{\tt hash}$.
We next prove $\widehat{B_1}.{\tt leaf} = \widehat{B_2}.{\tt leaf}$ and $\widehat{B_1}.{\tt leaf\_proof} = \widehat{B_2}.{\tt leaf\_proof}$.
Since $u_1$ has verified the Merkle proof before accepting $(B_1, \widehat{B_1})$, $u_1$ must have invoked $hash(x)$ for some $x$ of length $2\lambda$, with the return value of such invocation being $B_1.{\tt root}$.
Since ${\tt bad1}$ and ${\tt bad2}$ do not happen, such an $x$ must be unique. Let $x_0|x_1 \leftarrow x$, where $x_0$ and $x_1$ are both of length $\lambda$. Following this process for $\log_2 k$ steps (we are effectively tracing down the Merkle tree), we will find a unique $x_\bot$ value, which corresponds to leave $i$ of the Merkle tree, with $i$ being the last $\log_2 k$ bits of $\widehat{B_1}.{\tt hash}$. Since ${\tt bad1}$ and ${\tt bad2}$ do not happen, in order for $(B_1, \widehat{B_1})$ to pass the verification by $u_1$, we must have $\widehat{B_1}.{\tt leaf} = x_\bot$. By the same argument and since $B_1.{\tt root} = B_2.{\tt root}$, we must also have $\widehat{B_2}.{\tt leaf} = x_\bot$. Hence we conclude $\widehat{B_1}.{\tt leaf} = \widehat{B_2}.{\tt leaf}$. A similar argument shows that $\widehat{B_1}.{\tt leaf\_proof} = \widehat{B_2}.{\tt leaf\_proof}$ as well.

Next we move on to the ${\tt rank}$ and ${\tt next\_rank}$ fields. Let $A_1$ be any block on $u_1$ such that $\widehat{A_1}.{\tt hash} = x_\bot$, where $x_\bot$ is obtained as above. Similarly define $A_2$. Since ${\tt bad1}$ and ${\tt bad2}$ do not happen, we must have $A_1 = A_2$. Next, let $C_1$ be any block on $u_1$ such that $\widehat{C_1}.{\tt hash} = B_1.{\tt trailing}$, and let $C_2$ be any block on $u_2$ such that $\widehat{C_2}.{\tt hash} = B_2.{\tt trailing}= B_1.{\tt trailing}$. Similarly, we must also have $C_1 = C_2$.

Consider the DAG $\mathcal{G}_1$ consisting of all the blocks on $u_1$ as vertices, where for each block $B_1$, there is an edge to $B_1$ from the corresponding $A_1$, and another edge to $B_1$ from the corresponding $C_1$. We do a topological sort of all the vertices in $\mathcal{G}_1$, and assume that $B_1$ is the $j$th block in the topological sort. (Note that $j$ is based on $B_1$ and $\mathcal{G}_1$, and has nothing to do with $B_2$.) We will do an induction on $j$ to show that for all $B_1 = B_2$, we have $\widehat{B_1}.{\tt rank} = \widehat{B_2}.{\tt rank}$ and $\widehat{B_1}.{\tt next\_rank} = \widehat{B_2}.{\tt next\_rank}$.

The induction basis for $j$ from $1$ to $k$ trivially holds (these are the $k$ genesis blocks). Now assume that the previous claim holds for all $j < j_1$, and we prove the claim for $j = j_1$. Since $B_1$ is the $j_1$-th block in the topological sort, the position of $A_1$ and $C_1$ in the topological sort must be before $j_1$. Furthermore, we have shown earlier that $A_1 = A_2$ and $C_1 = C_2$. We can thus invoke the inductive hypothesis on $A_1$ and $C_1$, which shows that $\widehat{A_1}.{\tt next\_rank} = \widehat{A_2}.{\tt next\_rank}$ and
$\widehat{C_1}.{\tt next\_rank} = \widehat{C_2}.{\tt next\_rank}$.
By Line~\ref{code:rankstart} in Figure~\ref{fig:code},
$\widehat{B_1}.{\tt rank}$ is set to be the same as $\widehat{A_1}.{\tt next\_rank}$, while $\widehat{B_2}.{\tt rank}$ is set to be the same as $\widehat{A_2}.{\tt next\_rank}$. Hence we have $\widehat{B_1}.{\tt rank} = \widehat{B_2}.{\tt rank}$. A similar argument shows $\widehat{B_1}.{\tt next\_rank} = \widehat{B_2}.{\tt next\_rank}$.
\end{proof}

\begin{lemma}
\label{lemma:mainconsistency}
If the three properties in Theorem~\ref{the:blackbox} hold for each of the $k$ chains in $(k,p, \lambda, T)$-\tool, then with probability at least $1-exp(-\Omega(\lambda))$, $(k,p, \lambda, T)$-\tool satisfies the {\bf consistency} property in Theorem~\ref{the:main}.
\end{lemma}
\begin{proof}
Lemma~\ref{lemma:merkle} shows that with probability at least $1-exp(-\Omega(\lambda))$, there will never be two honest nodes adding $(B_1, \widehat{B_1})$ and $(B_2, \widehat{B_2})$ to their respective local set of blocks, such that $B_1 = B_2$ and $\widehat{B_1}\ne \widehat{B_2}$. This means that for each block $B$ accepted by an honest node, its ${\tt rank}$ and ${\tt next\_rank}$ on this honest node will be the same as the corresponding values on all other honest nodes. Hence we can directly refer to the ${\tt rank}$ and ${\tt next\_rank}$ of a block $B$, and no longer need to consider the values of $\widehat{B}.{\tt rank}$ and $\widehat{B}.{\tt next\_rank}$ on individual nodes. All our following discussion will be conditioned on this.

The following restates the {\bf consistency} property in Theorem~\ref{the:main}, which we need to prove:
\begin{itemize}
\item {\bf (consistency)} Consider the SCB $S_1$ on any node $u_1$ at any time $t_1$, and the SCB $S_2$ on any node $u_2$ at any time $t_2$.\footnote{Here $u_1$ ($t_1$) may or may not equal $u_2$ ($t_2$).} Then either $S_1$ is a prefix of $S_2$ or $S_2$ is a prefix of $S_1$. Furthermore, if ($u_1=u_2$ and $t_1 < t_2$) or
    ($u_1\ne u_2$ and $t_1 +\Delta < t_2$), then $S_1$ is a prefix of $S_2$.
\end{itemize}

Let the view of node $u_1$ at time $t_1$ be $\Psi_1$, and the view of node $u_2$ at time $t_2$ be $\Psi_2$. Let $x_1$ and $x_2$ be the ${\tt confirm\_bar}$ in $\Psi_1$ and $\Psi_2$, respectively. Without loss of generality, assume $x_1\le x_2$. Let $S_1$ and $S_2$ be the SCB in $\Psi_1$ and $\Psi_2$, respectively. The next will prove that $S_1$ is a prefix of $S_2$. In the proof, we will sometimes use set operations over $S_1$ and $S_2$. For example, $S_1\cap S_2$ refers to the set of common blocks in $S_1$ and $S_2$.

Let $F_1(i)$ be the sequence of partially-confirmed blocks on chain $i$ in $\Psi_1$. Let $G_1(i)$ be the prefix of $F_1(i)$ such that $G_1(i)$ contains all blocks in $F_1(i)$ whose ${\tt rank}$ is smaller than $x_1$. Similarly define $F_2(i)$ and $G_2(i)$, where $G_2(i)$ contains those blocks in $F_2(i)$ whose ${\tt rank}$ is smaller than $x_2$. We first prove the following two claims:
\begin{itemize}
\item For all $i$ where $0\le i\le k-1$, $G_1(i)$ is a prefix of $G_2(i)$. To prove this claim, note that by the {\bf consistency} property in Theorem~\ref{the:blackbox}, either $F_1(i)$ is a prefix of $F_2(i)$ or $F_2(i)$ is a prefix of $F_1(i)$. If $F_1(i)$ is a prefix of $F_2(i)$, then together with the fact that $x_1\le x_2$, it is obvious that $G_1(i)$ is a prefix of $G_2(i)$. If $F_2(i)$ is a prefix of $F_1(i)$, let $x_3$ be the ${\tt rank}$ of the last block in $F_2(i)$. This also means that the ${\tt next\_rank}$ of that block is at least $x_3+1$. By our design of ${\tt confirm\_bar}$, we know that $x_2\le x_3+1$. In turn, we have $x_1\le x_2\le x_3+1$. Hence $G_1(i)$ must also be a prefix of $G_2(i)$.
\item For all $i$ where $0\le i\le k-1$ and all block $B \in G_2(i)\setminus G_1(i)$, $B$'s ${\tt rank}$ must be no smaller than $x_1$. We prove this claim via a contradiction and assume that $B$'s ${\tt rank}$ is smaller than $x_1$. Together with the fact that $B$ is in $G_2(i)\setminus G_1(i)$, we know that $B$ is in $F_2(i)$ but not in $F_1(i)$. Hence $F_2(i)$ cannot be a prefix of $F_1(i)$. Then by the {\bf consistency} property in Theorem~\ref{the:blackbox}, $F_1(i)$ must be a prefix of $F_2(i)$.
    Let $x_4$ be the ${\tt rank}$ of the last block $D$ in $F_1(i)$.
    Since $F_1(i)$ is a prefix of $F_2(i)$, both $D$ and $B$ must be in $F_2(i)$, and $D$ must be before $B$ in $F_2(i)$. Since the blocks in $F_2(i)$ must have increasing ${\tt rank}$ values, we know that $D$'s ${\tt rank}$ must be smaller than $B$'s. Hence we have $x_4 = \mbox{$D$'s ${\tt rank}$} <
    \mbox{$B$'s ${\tt rank}$} \le x_1-1$, or more concisely, $x_4 < x_1-1$. On the other hand, since $D$'s ${\tt rank}$ is $x_4$, its ${\tt next\_rank}$ must be at least $x_4+1$.
    By our design of ${\tt confirm\_bar}$ in $\Psi_1$, we know that $x_1\le x_4+1$ and hence $x_4 \ge x_1-1$. This yields a contradiction.
\end{itemize}

Now we can use the above two claims to prove that $S_1$ is a prefix of $S_2$.
$S_1$ consists of all the blocks in $G_1(0)$ through $G_1(k-1)$, while $S_2$ consists of all the blocks in $G_2(0)$ through $G_2(k-1)$. Since $G_1(i)$ is a prefix of $G_2(i)$ for all $i$, we know that $S_1\subseteq S_2$. For all blocks in $S_1$ (which is the same as $S_2\cap S_1$), the sequence $S_1$ orders them in exactly the same way as the sequence $S_2$. For all block $B\in S_2\setminus S_1$, by the second claim above, we know that $B$'s ${\tt rank}$ must be no smaller than $x_1$. Hence in the sequence $S_2$, all blocks in $S_2\setminus S_1$ must be ordered after all the blocks in $S_2\cap S_1$ (whose ${\tt rank}$ must be smaller than $x_1$). This completes our proof that $S_1$ is a prefix of $S_2$.

Finally, if $u_1=u_2$ and $t_1 < t_2$, then since ${\tt confirm\_bar}$ on a node never decreases over time, we must have $x_1\le x_2$. Similarly, if $u_1\ne u_2$ and $t_1 +\Delta < t_2$, then by time $t_2$, node $u_2$ must have seen all blocks seen by $u_1$ at time $t_1$. By the {\bf consistency} property in Theorem~\ref{the:blackbox}, all partially-confirmed blocks on $u_1$ at time $t_1$ must also be partially-confirmed on $u_2$ at time $t_2$. Thus we must also have $x_1\le x_2$. Putting everything together, if ($u_1=u_2$ and $t_1 < t_2$) or ($u_1\ne u_2$ and $t_1 +\Delta < t_2$), then $S_1$ must be a prefix of $S_2$.
\end{proof}

\begin{lemma}
\label{lemma:mainqg}
If the three properties in Theorem~\ref{the:blackbox} hold for each of the $k$ chains in $(k,p, \lambda, T)$-\tool, then with probability at least $1-exp(-\Omega(\lambda))$, $(k,p, \lambda, T)$-\tool satisfies the {\bf quality-growth} property in Theorem~\ref{the:main}.
\end{lemma}
\begin{proof}
Same as the reasoning in the beginning of the proof for Lemma~\ref{lemma:mainconsistency}, we first invoke Lemma~\ref{lemma:merkle} to show that with probability at least $1-exp(-\Omega(\lambda))$, there will never be two honest nodes adding $(B_1, \widehat{B_1})$ and $(B_2, \widehat{B_2})$ to their respective local set of blocks, such that $B_1 = B_2$ and $\widehat{B_1}\ne \widehat{B_2}$. All our following discussion will be conditioned on this, and we will be able to directly refer to the ${\tt rank}$ and ${\tt next\_rank}$ of a block $B$.


The following restates the {\bf quality-growth} property in Theorem~\ref{the:main}, which we need to prove:
\begin{itemize}
\item {\bf (quality-growth)}
For all integer $\gamma\ge 1$, the following property holds after the very first $\frac{2T}{pn}$ ticks of the execution: On any honest node, in every
    $(\gamma +2)\cdot \frac{2T}{pn}+2\Delta$ ticks, at least
    $\gamma\cdot k \cdot  \frac{1-2f}{1-f}T$ honest blocks are newly added to SCB.
\end{itemize}

Consider any given honest node $u$, and any given time $t_0$ (in terms of ticks from the beginning of the execution), where $t_0$ is after the very first $\frac{2T}{pn}$ ticks of the execution. By the {\bf growth} property in Theorem~\ref{the:blackbox}, at time $t_0$, the length of every chain in $(k,p, \lambda, T)$-\tool must be at least $T$. By the {\bf growth} property and the {\bf quality} property in Theorem~\ref{the:blackbox}, we know that from time $t_0$ to time $t_1= t_0+ \gamma\cdot \frac{2T}{pn}$ on node $u$, every chain in $(k,p, \lambda, T)$-\tool has at least $\gamma\cdot \frac{1-2f}{1-f}T$ honest blocks becoming newly partially-confirmed. Let $\alpha_i$ denote the set of such newly partially-confirmed honest blocks on chain $i$ (for $0\le i\le k-1$). We have $|\alpha_i|\ge \gamma\cdot \frac{1-2f}{1-f}T$ for all $i$, and $|\cup_{i=0}^{k-1} \alpha_i| \ge \gamma\cdot k\cdot \frac{1-2f}{1-f}T$.
To prove the lemma, it suffices to show that by time $t_4 = t_0 + (\gamma+2)\cdot \frac{2T}{pn}+2\Delta$ on node $u$, all blocks in $\alpha_i$ will have become fully-confirmed for all $0\le i\le k-1$. Without loss of generality, we only need to prove for $\alpha_0$.


Let $x$ be the largest ${\tt next\_rank}$ among all the blocks in $\alpha_0$. By definition of such $x$, the ${\tt rank}$ value of every block in $\alpha_0$ must be smaller than $x$. Let $y$ be the length of chain $0$ on node $u$ at time $t_1$. By time $t_2 = t_1+ \Delta$, all honest nodes will have received all blocks in $\alpha_0$. Furthermore, the length of chain $0$ on all honest nodes at time $t_2$ must be at least $y$. Together with the {\bf consistency} property in Theorem~\ref{the:blackbox}, we know that all the blocks in $\alpha_0$ must be partially-confirmed on all honest nodes by time $t_2$.
Thus starting from $t_2$, whenever an honest node (including node $u$) mines a block, by the design of \tool, the ${\tt next\_rank}$ of the new honest block will be at least $x$. We will invoke this important property later.

Next let us come back to the honest node $u$, and consider any one of the chains.
From time $t_3 = t_2 + \Delta$ to time $t_4 = t_3 + \frac{4T}{pn}$ on node $u$, by the {\bf growth} property in Theorem~\ref{the:blackbox}, the length of this chain must have increased by at least $2T$ blocks. The first $T$ blocks among all these blocks must have been partially-confirmed on node $u$ at time $t_4$. By the {\bf quality} property in Theorem~\ref{the:blackbox}, these first $T$ blocks must contain at least
$\frac{1-2f}{1-f}T$ honest blocks. For $T\ge \frac{1-f}{1-2f}$, these first $T$ blocks must contain at least one honest block $B$ that is partially-confirmed. Since $B$ is first seen by $u$ no earlier than $t_3$, we know that this honest block $B$ must have been generated (either by $u$ or by some other honest node) no earlier than $t_2$. By our earlier argument, the ${\tt next\_rank}$ of $B$ must be at least $x$.

Finally, note that we actually have one such $B$ on {\em every} chain on node $u$ at time $t_4$. This means that on node $u$ at time $t_4$, the ${\tt confirm\_bar}$ is at least $x$. We earlier showed that for all blocks in $\alpha_0$, their ${\tt rank}$ values must all be smaller than $x$. Hence such a ${\tt confirm\_bar}$ enables all blocks in $\alpha_0$ to become fully-confirmed on node $u$ at time $t_4$. Observe that $t_4 = t_0 + (\gamma+2)\cdot \frac{2T}{pn}+2\Delta$, and we are done.
\end{proof}

\vspace*{3mm}
\subsection{Additional Formalism for the Proof of Lemma~\ref{lemma:reduction}}
\label{app:reductionformal}
\vspace*{3mm}

This section introduces some additional modelling and formalism, to prepare for the proof of Lemma~\ref{lemma:reduction} in Appendix~\ref{app:reductionproof}.

\Paragraph{PoW based on a random oracle.}
In \tool, for performing and verifying PoW, all the nodes have access to  a common random orale ${\tt H1}$ (i.e., the hash function, which is viewed as a random oracle).
We need to properly model PoW based on such a random oracle, and we use exactly the same approach as in Pass et al.'s analysis on Nakamoto consensus~\cite{pass2017analysis}. Specifically, let the function $f()$ be the random function corresponding to ${\tt H1}$. The random oracle ${\tt H1}$ has two interfaces: For mining a block, an honest node invokes ${\tt H1.compute}(B)$ to compute $f(B)$ (i.e., Line~\ref{code:h1compute} of Figure~\ref{fig:code} will become ``$\widehat{B}.{\tt hash} \leftarrow {\tt H1.compute}(B)$''). Such mining is successful iff $f(B)$ has a certain number of leading zeroes. For verifying a received block, an honest node invokes ${\tt H1.verify}(B,y)$ to check whether $f(B)$ equals $y$ (i.e., Line~\ref{code:h1verify} of Figure~\ref{fig:code} will become
``verify that ${\tt H1.verify}(B,\widehat{B}.{\tt hash}) = {\tt true}$''). The adversary may invoke ${\tt H1.compute}(B)$ and ${\tt H1.verify}(B,y)$ in {\em arbitrary} ways --- for example, the adversary may choose not to invoke ${\tt H1.verify}(B,y)$, but instead invoke ${\tt H1.compute}(B)$ to verify whether $f(B)=y$. In each tick, an honest node can invoke ${\tt H1.compute}()$ only once, and the adversary can invoke ${\tt H1.compute}()$ at most $fn$ times. On the other hand, ${\tt H1.verify}()$ invocations are ``free'' and can be invoked any number of times (except that the total number of such invocations is still bounded by length of the entire execution).
As explained in \cite{pass2017analysis}, separating these two interfaces of ${\tt H1}$ helps to capture the assumption that doing a PoW incurs non-trivial cost, while verifying a PoW does not.

Related to the above, and also following the same approach in Pass et al.~\cite{pass2017analysis}, in \tool each block $B$ is sent together with its corresponding $f(B)$ value. Specifically, the attachment $\widehat{B}$ contains a field $\widehat{B}.{\tt hash}$ to store $f(B)$. This is needed so that the receiver can verify $B$ by invoking the ``free'' interface ${\tt H1.verify}(B,\widehat{B}.{\tt hash})$, instead of invoking ${\tt H1.compute}(B)$ to recompute $f(B)$.

\Paragraph{Random oracle used in Merkle tree.}
In \tool, the nodes also use a hash function to construct and verify the Merkle tree. It will be convenient to model this hash function as a separate random oracle ${\tt H2}$ that is independent of ${\tt H1}$.\footnote{Note that the leaves on our Merkle tree are hashes of blocks. The hashes of the blocks are obtained using ${\tt H1}$, while the non-leave nodes on the Merkle tree are computed using ${\tt H2}$.} Since a single random oracle can be easily used to implement two independent random oracles, doing so does not make any difference in implementation. The oracle ${\tt H2}$ also has two interfaces, ${\tt H2.compute}(x)$ and ${\tt H2.verify}(x,y)$. An honest node invokes ${\tt H2.compute}(x)$ for constructing the Merkle tree (Line~\ref{code:h2computeinitial} and Line~\ref{code:h2compute} of Figure~\ref{fig:code}), and invokes ${\tt H2.verify}(x,y)$ for verifying the Merkle tree (Line~\ref{code:h2verify} of Figure~\ref{fig:code}). The adversary may invoke these two interfaces in {\em arbitrary} ways --- for example, it may choose to invoke ${\tt H2.compute}(x)$ for verifying the Merkle proof. All invocations of ${\tt H2.compute}(x)$ and ${\tt H2.verify}(x,y)$ are ``free''. Separating these two interfaces helps to make our discussion less cumbersome --- for example, we can refer to ``the invocation of the random oracle ${\tt H2}$ at Line~\ref{code:h2verify} of Figure~\ref{fig:code}'' simply as ``${\tt H2.verify}()$''.

\begin{figure}
\begin{subfigure}{.48\textwidth}
\pgfplotsset{compat=newest,
legend style={font=\footnotesize},
label style={font=\footnotesize},
tick label style={font=\footnotesize},
title style={font=\footnotesize}}
\centering\captionsetup{width=.8\linewidth}
\begin{framed}
\footnotesize\small
\begin{algorithmic}[1]
\State $V' \leftarrow$ \{(genesis block, genesis block's attachment)\};
\State ${\tt last'} \leftarrow$ hash of the genesis block;
\State
\State Nakamoto() \{
\State \hspace*{2mm}{\bf repeat forever} \{
\State \hspace*{8mm}ReceiveState();
\State \hspace*{8mm}Mining();
\State \hspace*{8mm}SendState();
\State \hspace*{2mm}\}
\State \}
\State
\State Mining() \{
\State \hspace*{2mm}$B'.{\tt transactions'} \leftarrow$ get\_transactions();
\State \hspace*{2mm}$B'.{\tt prev'} \leftarrow {\tt last'}$;
\State \hspace*{2mm}$B'.{\tt nonce'} \leftarrow$  new\_nonce();
\State \hspace*{2mm}$\widehat{B'}.{\tt hash'} \leftarrow {\tt H1'.compute}(B')$;
\label{code:h1pcompute}
\State \hspace*{2mm}{\bf if} ($\widehat{B'}.{\tt hash'}$ has $\log_2 (\frac{1}{kp})$ leading zeroes and $\log_2 k$ trailing zeroes)
{\bf then} ProcessBlock($B'$, $\widehat{B'}$);
\State \}
\State
\State ProcessBlock($B'$, $\widehat{B'}$) \{
\State \hspace*{2mm}verify that $\widehat{B'}.{\tt hash'}$ has $\log_2 (\frac{1}{kp})$ leading zeroes and $\log_2 k$ trailing zeroes;
\State \hspace*{2mm}verify that ${\tt H1'.verify}(B', \widehat{B'}.{\tt hash'})= {\tt true}$;
\label{code:h1pverify}
\State \hspace*{2mm}verify that $B'.{\tt prev'} = \widehat{A'}.{\tt hash'}$
for some block $(A', \widehat{A'}) \in V'$;
\State \hspace*{2mm} {\bf if} (any of above 3 verifications fails) {\bf then return};
\State
\State \hspace*{2mm}$V' \leftarrow V' \cup \{(B', \widehat{B'})\}$;
\State \hspace*{2mm}${\tt last'} \leftarrow \mbox{hash of the last block on the longest path in $V'$}$;
\State \}
\State
\State SendState() \{
\State \hspace*{2mm} send $V'$ to other nodes;
\State \hspace*{2mm} // In implementation, only need to send those blocks not sent before.
\State \}
\State
\State ReceiveState() \{
\State \hspace*{2mm} {\bf foreach} ($B'$, $\widehat{B'}$) $\in$ received state {\bf do} \State \hspace*{6mm} ProcessBlock($B'$, $\widehat{B'}$);
\State // Note that a block $B'_1$ should be processed before $B'_2$ if
$B'_2.{\tt prev'}$ points to $B'_1$.
\State \}
\end{algorithmic}
\normalsize
\end{framed}
\end{subfigure}
\vspace*{2mm}
\caption{Pseudo-code of Nakamoto consensus.}
\label{fig:nakacode}
\vspace*{-3mm}
\end{figure}

\Paragraph{Pseudo-code for Nakamoto consensus.}
To make our reduction in Lemma~\ref{lemma:reduction} rigorous, we specify the Nakamoto consensus protocol in Figure~\ref{fig:nakacode} --- this specification is the same as in Pass et al.'s work~\cite{pass2017analysis}. Several aspects of the specification are worth some further clarification.

First, as explained above, we use ${\tt H1'.compute}()$ and ${\tt H1'.verify}()$ to model random-oracle-based PoW (Line~\ref{code:h1pcompute} and Line~\ref{code:h1pverify} in Figure~\ref{fig:nakacode}). To avoid notation collision, we use ${\tt H1'}$ to refer to the random oracle in Nakamoto consensus.

Second, exactly the same as in Pass et al.'s work~\cite{pass2017analysis}, in Figure~\ref{fig:nakacode}, each block $B'$ in Nakamoto consensus is sent together with its hash value, which is stored in the attachment $\widehat{B'}$. (Hence each block $B'$ in Nakamoto consensus now also has an attachment.) As explained earlier, doing so serves to enable the receiver to verify a block by invoking the ``free'' interface ${\tt H1'.verify}()$.

Finally, usually in Nakamoto consensus, a valid block's hash needs to have $\log_2 \frac{1}{p}$ leading zeroes. To simplify discussion later, we instead require a valid block's hash to have $\log_2 (\frac{1}{kp})$ leading zeroes and $\log_2 k$ trailing zeroes. Since we model hash functions as random oracles and since $\log_2 (\frac{1}{kp}) + \log_2 k = \log_2 \frac{1}{p}$, such a change does not make any material difference for the protocol.

\Paragraph{Modelling adversarial message delay.}
The adversary can manipulate the message delay in arbitrary ways, as long as the delay does not exceed $\Delta$. The adversary also sees every message as soon as it is sent. To model this, we follow exactly the same approach as in \cite{pass2017analysis}: We imagine that honest nodes do not communicate with each other directly. Instead, an honest node always directly sends its message to the adversary (with no delay). The adversary then must deliver this message to all other honest nodes within $\Delta$ ticks, without modifying the content of the message.

\vspace*{3mm}
\subsection{Overview of Proof for Lemma~\ref{lemma:reduction}}
\label{app:reductionoverview}
\vspace*{3mm}

Since the proof for Lemma~\ref{lemma:reduction} in Appendix~\ref{app:reductionproof}
is long, we first give an overview here.
Recall from Appendix~\ref{app:reductionformal} that \tool uses two random oracles ${\tt H1}$ and ${\tt H2}$, and they are invoked via interfaces ${\tt H1.compute}()$, ${\tt H1.verify}()$, ${\tt H2.compute}()$, and ${\tt H2.verify}()$.

The first conceptual step in the proof of Lemma~\ref{lemma:reduction} is to clean up various low-probability bad events such as collisions, so that the execution (of \tool against the given adversary $\mathcal{A}$) becomes easier to reason about. Specifically, related to ${\tt H1}$, there are two bad events:
\begin{itemize}
\item The execution contains some invocation ${\tt H1.verify}(x, y)$ such that ${\tt H1.compute}(x)$ has not been previously invoked and yet ${\tt H1.verify}(x, y)$ returns ${\tt true}$.
\item The execution contains some invocations ${\tt H1.compute}(x_1)$ and ${\tt H1.compute}(x_2)$ such that $x_1\ne x_2$ and ${\tt H1.compute}(x_1) = {\tt H1.compute}(x_2)$.
\end{itemize}
The first step in the proof of Lemma~\ref{lemma:reduction} modifies the semantics of  ${\tt H1.verify}(x, y)$ and ${\tt H1.compute}(x)$, so that the above two bad events will never happen. The proof will further modify the semantics of  ${\tt H2.verify}(x, y)$ and ${\tt H2.compute}(x)$ in a similar way. We will prove that all these modifications only change the distribution of the execution by some exponentially small amount.

The next conceptual step is the central part of the proof. Without loss of generality, let us consider $i=0$ in Lemma~\ref{lemma:reduction}, and we simply use $\mathcal{A}'$ to denote $\mathcal{A}'_0$.
For any given adversary $\mathcal{A}$ for $(k,p, \lambda, T)$-\tool, we need to construct another adversary $\mathcal{A}'$ for $(p, \lambda, T)$-Nakamoto such that chain $0$ in the execution of $(k,p, \lambda, T)$-\tool against $\mathcal{A}$ behaves almost exactly the same as the chain in the execution of $(p, \lambda, T)$-Nakamoto against $\mathcal{A}'$.

To facilitate understanding, all notations for the execution of $(p, \lambda, T)$-Nakamoto against $\mathcal{A}'$ will come with the superscript ``$'$'', while notations for the execution of $(k,p, \lambda, T)$-\tool will not. Let $u'$ denote any honest node in $(p, \lambda, T)$-Nakamoto. We construct $\mathcal{A}'$ in the following way: $\mathcal{A}'$ is composed of a {\em middlebox} and the adversary $\mathcal{A}$ --- putting it another way, $\mathcal{A}'$ will invoke $\mathcal{A}$ as a black-box. The honest nodes in $(p, \lambda, T)$-Nakamoto only interact with the middlebox. The middlebox internally simulates an execution of $(k,p, \lambda, T)$-\tool running against $\mathcal{A}$, by simulating all the \tool honest nodes and by invoking $\mathcal{A}$. We will use $u$ to denote any (simulated) honest node in this simulated execution.

The design of the middlebox is the key. In each tick, to simulate the honest node $u$ (which runs \tool), the middlebox observes the outgoing message from the corresponding honest node $u'$ (which runs Nakamoto):
\begin{itemize}
\item If $u'$ sends out a new block $B'$, it means that $u'$ successfully mined $B'$ in this tick. The middlebox then simulates the mining by $u$. When $u$ invokes ${\tt H1.compute}()$, the middlebox directly returns $\widehat{B'}.{\tt hash}$ as the result. Since $\widehat{B'}.{\tt hash}$ is already known to have $\log_2 (\frac{1}{kp})$ leading zeroes and $\log_2 k$ trailing zeroes, this will cause $u$ to successfully mine a block $B$, and $B$ will eventually be assigned to chain $0$. The middle box records the mapping from $B'$ to $B$ in a table, and records the mapping from $B$ to the query result (i.e., $\widehat{B'}.{\tt hash}$) in another table.
\item If $u'$ does not send out any message, it means that $u'$ has failed to mine a block in this tick.
The middlebox then simulates the mining by $u$. When $u$ invokes ${\tt H1.compute}()$, the middlebox returns a uniformly random value $y$, conditioned upon that $y$ does not correspond to successful mining in Nakamoto. (In other words, $y$ should either not have $\log_2 (\frac{1}{kp})$ leading zeroes or not have $\log_2 k$ trailing zeroes.) Finally, the middlebox does proper bookkeeping as earlier.
\end{itemize}

The middle-box also needs to interact with the simulated adversary $\mathcal{A}$:
\begin{itemize}
\item The middle box feeds all messages sent by the simulated honest nodes into $\mathcal{A}$.
\item If $\mathcal{A}$ want to deliver a message to a simulated honest node $u$, the middlebox feeds that message into $u$, and observes whether $u$ add any block $B$ to its chain $0$. If yes, the middlebox converts $B$ (which is an \tool block) to $B'$ (which is a Nakamoto block), and delivers $B'$ to the corresponding $u'$.
\item If $\mathcal{A}$ invokes ${\tt H1.compute}(B)$ (i.e., the PoW query), the middlebox converts $B$ (which is an \tool block) to $B'$ (which is a Nakamoto block). The middlebox then ``tunnels'' this request to the random oracle ${\tt H1}'$ in the Nakamoto execution, by invoking ${\tt H1'.compute}(B')$ and then relaying the response to $\mathcal{A}$. (This invocation of ${\tt H1'.compute}()$ will be ``charged'' against $\mathcal{A}'$, of which the middlebox is a part.) Finally, the middlebox does proper bookkeeping.
\item All invocations of ${\tt H1.verify}()$, ${\tt H2.compute}()$, and ${\tt H2.verify}()$ by $\mathcal{A}$ are answered by the middlebox looking up its internal tables.
\end{itemize}

Finally, one of the tables maintained by the middlebox will directly give us the mapping $\sigma_0^\tau$ as needed by Lemma~\ref{lemma:reduction}.

\vspace*{3mm}
\subsection{Proof for Lemma~\ref{lemma:reduction}}
\label{app:reductionproof}
\vspace*{3mm}

\begin{proof} (for Lemma~\ref{lemma:reduction})
Without loss of generality, we prove this lemma for $i=0$. To simplify notation, we use $\mathcal{A}'$ to denote $\mathcal{A}'_0$ in the lemma.
Let ${\tt H1}$ and ${\tt H2}$ be the two random oracles used in $\tt{EXEC}(\mbox{\tool}, k,p, \lambda, T, \mathcal{A})$, and ${\tt H1'}$ be the (single) random oracle used in $\tt{EXEC}(\mbox{Nakamoto}, p, \lambda, T, \mathcal{A}')$. We will design and reason about several hybrid executions. For convenience, define  ${\tt Hybrid1}$ to be $\tt{EXEC}(\mbox{\tool}, k,p, \lambda, T, \mathcal{A})$.

\Paragraph{Designing ${\tt Hybrid2}$.}
%
We define ${\tt Hybrid2}$ to be the same as ${\tt Hybrid1}$, except that in ${\tt Hybrid2}$ we change the semantics of all invocations to ${\tt H1.compute}()$, ${\tt H1.verify}()$, ${\tt H2.compute}()$, and ${\tt H2.verify}()$.
To do so, we maintain three tables $T_1$, $T_2$, and $T_3$. The three tables contain entries in the form of $(B, hash)$, $(x,y)$, and $(B, B')$, respectively.
Let the $k$ genesis blocks of \tool be $(B_0, \widehat{B_0}), (B_1, \widehat{B_1}), ... , (B_{k-1}, \widehat{B_{k-1})}$, and let the genesis block of Nakamoto be $(B', \widehat{B'})$. Initially, $T_1$ contains entries $(B_{i}, \widehat{B_{i}}.\tt{hash})$ for each $i \in [0, k-1]$, $T_2$ is empty, and $T_3$ contains the entry $(B_{0},B')$.

\begin{itemize}
\item In ${\tt Hybrid2}$, when ${\tt H2.verify}(x,y)$ is invoked, we return ${\tt true}$ if $(x,y)\in T_2$, and ${\tt false}$ otherwise.

We define ${\tt bad1}$ to be the event that the execution
${\tt Hybrid1}$ contains some invocation ${\tt H2.verify}(x, y)$ such that ${\tt H2.compute}(x)$ has not been previously invoked and yet ${\tt H2.verify}(x, y)$ returns ${\tt true}$. Given that ${\tt H2}$ is a random oracle, and given that the length of the execution ${\tt Hybrid1}$ is $\mbox{poly}(\lambda)$, we have
$\Pr[{\tt bad1}] = exp(-\Omega(\lambda))$.

\item In ${\tt Hybrid2}$, when ${\tt H2.compute}(x)$ is invoked, if there exists some $y$ such that $(x,y)\in T_2$, we directly return $y$. (By our design, such a $y$ must be unique.) Otherwise we choose a uniformly random binary string $y$ of $\lambda$ length. We next check whether $T_2$ contains any entry in the form of ($x_2$, $y$) for any $x_2$. If yes, we re-choose a uniformly random $y$, until there is no collision.
    We next insert $(x,y)$ into $T_2$, and returns $y$ to the invoking party.

 We define ${\tt bad2}$ to be the event that the execution ${\tt Hybrid1}$ contains some invocations ${\tt H2.compute}(x_1)$ and ${\tt H2.compute}(x_2)$ such that $x_1\ne x_2$ and ${\tt H2.compute}(x_1) = {\tt H2.compute}(x_2)$. By similar reasoning as above, we have $\Pr[{\tt bad2}] = exp(-\Omega(\lambda))$.

\item In ${\tt Hybrid2}$, when ${\tt H1.verify}(B,y)$ is invoked, we return ${\tt true}$ if $(B, y) \in T_1$, and ${\tt false}$ otherwise.

    We define ${\tt bad3}$ to be the event that the execution
${\tt Hybrid1}$ contains some invocation ${\tt H1.verify}(B, y)$ such that ${\tt H1.compute}(B)$ has not been previously invoked and yet ${\tt H1.verify}(B, y)$ returns ${\tt true}$. By similar reasoning as above, we have $\Pr[{\tt bad3}] = exp(-\Omega(\lambda))$.

\item In ${\tt Hybrid2}$, when ${\tt H1.compute}(B)$ is invoked, if there exists some $y$ such that $(B,y)\in T_1$, we directly return $y$. Otherwise
we check whether there is some $B'$ such that $(B, B')\in T_3$. If not, we construct $B'$ based on $B$ (as described next), and add $(B, B')$ into $T_3$. We then invoke ${\tt H1'.compute}(B')$, and let the return value be $y$. We add $(B, y)$ to $T_1$, and return $y$ as the answer to ${\tt H1.compute}(B)$.

To construct $B'$ based on $B$, we first set
$B'.{\tt transactions'}\leftarrow B.{\tt transactions}$,
and set $B'.{\tt nonce'}$ to be a uniformly randomly chosen nonce. Next
we check whether $T_2$ contains any entry in the form of $(x_0|x_1, B.{\tt root})$ for some $x_0$ and $x_1$, each of length $\lambda$. (Given that $T_2$ contains no collisions, $x_0$ and $x_1$ must be unique, if they do exist.) If yes, we continue to look for an entry in the form of $(x_{00}|x_{01}, x_0)$, for some $x_{00}$ and $x_{01}$ each of length $\lambda$. We repeat such a step for $\log_2 k$ times, to obtain $x_0$, $x_{00}$, $x_{000}$, $\ldots$. Let $x_\bot$ be the last value in the previous sequence of $\log_2 k$ values, and we set $B'.{\tt prev} \leftarrow x_\bot$.
If any of the above $\log_2 k$ steps cannot proceed due to $T_2$ not containing the needed entry,  we set $B'.{\tt prev} \leftarrow B.{\tt root}$.

We define ${\tt bad4}$ to be the event that in the execution of ${\tt Hybrid2}$, the table $T_3$ contains two entries $(B_1, B'_1)$ and $(B_2, B'_2)$ such that $B_1 = B_2$ or $B'_1 = B'_2$. Since each $B'$ comes with a new nonce and since the length of the execution is $\mbox{poly}(\lambda)$, we have $\Pr[{\tt bad4}] = exp(-\Omega(\lambda))$.
\end{itemize}

\Paragraph{Execution ${\tt Hybrid2}$ is close to ${\tt Hybrid1}$.}
Conditioned upon none of the four bad events (${\tt bad1}$ through ${\tt bad4}$) happening, one can verify that the joint distribution of the return values of all the invocations of the four interfaces (i.e., ${\tt H1.compute}()$, ${\tt H1.verify}()$, ${\tt H2.compute}()$, and ${\tt H2.verify}()$) is the same in ${\tt Hybrid1}$ and ${\tt Hybrid2}$. Since by a union bound, the probability of none of the four bad events happening is at least $1- exp(-\Omega(\lambda))$, we conclude that ${\tt Hybrid1}$ and ${\tt Hybrid2}$ are strongly statistically close.

\Paragraph{Designing ${\tt Hybrid3}$.}
We define ${\tt Hybrid3}$ to be the same as ${\tt Hybrid2}$, except that i)
the honest nodes in ${\tt Hybrid3}$ run $(p, \lambda, T)$-Nakamoto, rather than the protocol in ${\tt Hybrid2}$, and ii) the adversary $\mathcal{A}'$ in ${\tt Hybrid3}$ is composed of a special {\em middlebox} and the adversary $\mathcal{A}$ in ${\tt Hybrid2}$. The middle-box converts i) incoming messages to $\mathcal{A}'$ so that they can be fed into $\mathcal{A}$, and ii) outgoing messages from $\mathcal{A}$ so that they can be sent by $\mathcal{A}'$ (to the honest nodes running Nakamoto consensus). Furthermore, all the honest nodes in ${\tt Hybrid2}$ will be simulated by the middlebox in ${\tt Hybrid3}$.
Consider any given honest node $u'$ in ${\tt Hybrid3}$, which runs $(p, \lambda, T)$-Nakamoto and maintains a single chain. The middlebox internally simulates a counterpart of $u'$, and let us call this simulated counterpart as node $u$. This (simulated) node $u$ runs the protocol in ${\tt Hybrid2}$, and hence maintains $k$ parallel chains. We will next explain how the middlebox exactly works. To avoid notation collision, we will use $u'$ to denote any honest node in ${\tt Hybrid3}$, and $u$ to denote any honest node in ${\tt Hybrid2}$ and any (simulated) honest node in ${\tt Hybrid3}$ as simulated by the middlebox.

In each tick, if node $u'$ broadcasts a new block ($B'$, $\widehat{B'}$), then:
\begin{itemize}
\item The middle-box simulates the execution of Mining() by $u$. When $u$ invokes ${\tt H1.compute}(B)$, the middle-box will return $\widehat{B'}.{\tt hash}$ as the answer to ${\tt H1.compute}(B)$.
\item The middle-box adds $(B, \widehat{B}.{\tt hash})$ into the table $T_1$, where $(B, \widehat{B}.{\tt hash})$ was constructed by $u$ during Mining().
\item The middle-box adds $(B, B')$ into the table $T_3$.
\item The middle-box feeds the message sent by $u$ (if any) into $\mathcal{A}$.
\end{itemize}

In each tick, if node $u'$ does not broadcast a new block, then:
\begin{itemize}
\item The middle-box simulates the execution of Mining() by $u$. When $u$ invokes ${\tt H1.compute}(B)$, the middle-box returns $y$ as the answer. Here $y$ is chosen as a uniformly random binary string of length $\lambda$. But if $y$ has $\log_2 (\frac{1}{kp})$ leading zeroes and $\log_2 k$ trailing zeroes, the middle-box will re-choose $y$ uniformly randomly, until $y$ does not have such a property.
\item The middle-box adds $(B, \widehat{B}.{\tt hash})$ into the table $T_1$, where $(B, \widehat{B}.{\tt hash})$ was constructed by $u$ during Mining().
\item The middlebox next constructs $B'$ based on $B$, in the same process as
when ${\tt H1.compute}(B)$ is invoked in ${\tt Hybrid2}$. The middlebox adds $(B, B')$ into the table $T_3$.
\item The middle-box feeds the message sent by $u$ (if any) into $\mathcal{A}$.
\end{itemize}

Recall that in ${\tt Hybrid3}$, $\mathcal{A}$ is being used as a component by  $\mathcal{A}'$. In ${\tt Hybrid3}$, all invocations of ${\tt H1.compute}()$, ${\tt H1.verify}()$, ${\tt H2.compute}()$, and ${\tt H2.verify}()$ by $\mathcal{A}$ are processed in the same way in the ${\tt Hybrid2}$.
In each tick, if $\mathcal{A}$ (as part of $\mathcal{A}'$) intends to deliver a message to node $u$, then for each tuple $(B, \widehat{B})$ in the message:
\begin{itemize}
\item The middlebox simulates the execution of ProcessBlock$(B, \widehat{B})$ by node $u$.
\item If $u$ adds $(B, \widehat{B})$ to $V_0$, the middlebox will find $B'$ such that $(B, B') \in T_3$. The middlebox sets $\widehat{B'}.{\tt hash} \leftarrow \widehat{B}.{\tt hash}$, and forwards/delivers the tuple $(B', \widehat{B'})$ to node $u'$.
    Note that for this step, there must exists some $B'$ such that $(B, B') \in T_3$. The reason is that $B$ already passed the verification by $u$, and part of that verification process invoked
    ${\tt H1.verify}(B, \widehat{B}.{\tt hash})$ with a return value of ${\tt true}$.
    Hence we have $(B, \widehat{B}.{\tt hash})\in T_1$. One can easily verify that whenever the execution ${\tt Hybrid3}$ adds $(B, \widehat{B}.{\tt hash})$ into the table $T_1$, some entry $(B, B')$ must be added into the table $T_3$.
\end{itemize}

\Paragraph{Properties of ${\tt Hybrid3}$.}
Define ${\tt bad5}$ to be the event where in ${\tt Hybrid3}$, the table $T_3$ contains two entries $(B_1, B'_1)$ and $(B_2, B'_2)$ such that either $B_1 = B_2$ or $B'_1 = B'_2$. In ${\tt Hybrid3}$, a new entry $(B, B')$ is added to $T_3$ only when some honest node $u$ invokes ${\tt H1.compute}()$ or $\mathcal{A}$ invokes ${\tt H1.compute}()$. When $u$ invokes ${\tt H1.compute}()$, both the nonce in $B$ and the nonce in $B'$ are uniformly random. When $\mathcal{A}$ invokes ${\tt H1.compute}()$, a new entry $(B, B')$ is added to $T_3$ only if $T_3$ does not already contain an entry for $B$. In such a case, $B'$ in the new entry $(B, B')$ always has a random nonce. Putting both cases together and since the length of the execution is $\mbox{poly}(\lambda)$, we have $\Pr[{\tt bad5}] = exp(-\Omega(\lambda))$.


It is easy to see that ${\tt Hybrid3}$ is just an execution of the $(p, \lambda, T)$-Nakamoto protocol against the adversary $\mathcal{A}'$ (which is composed of $\mathcal{A}$ and the middlebox).
Furthermore, one can (tediously) verify that conditioned upon neither ${\tt bad4}$ nor ${\tt bad5}$ happening, the view of $\mathcal{A}$ (i.e., the randomness in $\mathcal{A}$, all the incoming messages to $\mathcal{A}$, and all the return values for invocations of ${\tt H1}$ and ${\tt H2}$ by $\mathcal{A}$) in ${\tt Hybrid2}$ follows the same distribution as the view of $\mathcal{A}$ in ${\tt Hybrid3}$. Hence conditioned upon neither ${\tt bad4}$ nor ${\tt bad5}$ happening, the behavior of $\mathcal{A}$ in ${\tt Hybrid2}$ follows the same distribution as the behavior of $\mathcal{A}$ in ${\tt Hybrid3}$. In turn, conditioned upon neither ${\tt bad4}$ nor ${\tt bad5}$ happening, the state of all the honest nodes $u$ in  ${\tt Hybrid2}$ is identically distributed as the state of all the (simulated) honest nodes $u$ in  ${\tt Hybrid3}$.

Finally, conditioned upon the event ${\tt bad5}$ not happening, the table $T_3$ in ${\tt Hybrid3}$ gives a one-to-one mapping $\sigma_0^\tau$ from each block $B'$ in all the $V'$ on all the nodes $u'$ in ${\tt Hybrid3}$ to some block $\sigma_0^\tau(B')$ in all the $V_i$ on all the (simulated) nodes $u$ in ${\tt Hybrid3}$. Furthermore, $\sigma_0^\tau$ guarantees that $B'$ is an honest block iff $\sigma_0^\tau(B')$ is an honest block, and also guarantees that $B'.{\tt prev} = \widehat{\sigma_0^\tau(B')}.{\tt leaf}$.
One can verify that in ${\tt Hybrid3}$, a node $u'$ adds a block $B'$ to its $V'$ iff the (simulated) node $u$ adds the block $\sigma_0^\tau(B')$ to $V_0$.

The following summarizes what we have by now, conditioned upon neither ${\tt bad4}$ nor ${\tt bad5}$ happening:
\begin{itemize}
\item ${\tt Hybrid1} = \tt{EXEC}(\mbox{\tool}, k,p, \lambda, T, \mathcal{A})$.
\item ${\tt Hybrid1}$ and ${\tt Hybrid2}$ are strongly statistically close. (This holds even conditioned upon ${\tt bad4}$ not happening.)
\item The joint state of all the honest nodes $u$ in  ${\tt Hybrid2}$ is identically distributed as the joint state of all the simulated honest nodes $u$ in  ${\tt Hybrid3}$.
\item An honest node $u'$ in ${\tt Hybrid3}$ adds $B'$ to its $V'$ iff the corresponding simulated honest node $u$ in ${\tt Hybrid3}$  adds $\sigma_0^\tau(B')$ to its $V_0$, and we must further have $B'.{\tt prev} = \widehat{\sigma_0^\tau(B')}.{\tt leaf}$.
\item $B'$ is an honest block iff $\sigma_0^\tau(B')$ is an honest block.
\item ${\tt Hybrid3}= \tt{EXEC}(\mbox{Nakamoto}, p, \lambda, T, \mathcal{A}')$.
\end{itemize}
Combining all the above then completes the proof for the lemma.
\end{proof}

\vspace*{3mm}
\subsection{Security Properties of Conflux}
\label{sec:conflux}
\vspace*{3mm}

This section presents concrete results showing that in our simulation,
as the throughput of Conflux~\cite{conflux} increases, the security properties of Conflux deteriorate. Our simulation results are based on the balance attack, which is related to the one in \cite{bagaria18,natoli16}. We do note, however, that our simulation results are obtained only under some example parameters from the Conflux paper~\cite{conflux}. These results may change under other parameters.

We do not claim novelty or research contribution in this section, since the balance attack on Conflux has been folklore and is not the focus of this paper. Instead, our goal is simply to show, concretely, how the security properties of Conflux deteriorate in our simulation. Related observations on Conflux have also been independently made by Bagaria et al.~\cite{bagaria18} and Fitzi et al.~\cite{fitzi18}.

\Paragraph{Possible fix.}
A simple and easy fix to strengthen Conflux's security properties against the balance attack is for Conflux to operate closer to the parameter range of the GHOST protocol~\cite{ghost}, with relatively ``low'' throughput. (Such relatively ``low'' throughput will still be much higher than the throughput of the GHOST protocol.) In fact, our simulation, under the ``low'' throughput setting, already confirms the effectiveness of this fix.

Another good solution for improving Conflux's security properties against the balance attack is for Conflux to explicitly use multiple chains and to explicitly force the adversary to split its computational power across these multiple chains, as in \cite{fitzi18}. Other fixes may also exist, which is however beyond the scope of this paper.

\Paragraph{Review of Conflux.}
Conflux is based on the GHOST protocol~\cite{ghost}, while GHOST is an improvement over the Bitcoin protocol.
Imagine for now that the Bitcoin protocol disseminates all blocks that have ever been mined, regardless of whether they are on the longest path.
Let $\Psi$ be the set of blocks that any given node have seen by any given point of time. Then all blocks in $\Psi$ will conceptually form a direct tree, where block $A$ has a directed edge to block $B$ if $A$ extends from $B$. With a slight abuse of notation, we also use $\Psi$ to refer to this tree. In the Bitcoin protocol, to produce the total ordering of confirmed blocks, a node will follow the longest-path rule. Namely, the node will output the longest path (after truncating a certain number of blocks at the end) from the root to some leaf of the tree, where the leaf is chosen such that the path length is maximized.

In GHOST, rather than choosing the longest path, a node instead chooses the path in the following way. The node will start from the root $A$ of $\Psi$, and then find $A$'s child $B$ such that $B$'s subtree has the most blocks, as compared to the subtrees for other children of $A$. (Namely, $B$'s subtree is the ``heaviest''.) The path now contains $AB$. Next, the node finds $B$'s child $C$ such $C$'s subtree is the heaviest, among the subtrees for all children of $B$. The path now contains $ABC$. The process continues until we reach a leaf.

Conflux builds upon GHOST by introducing a mechanism for including all blocks in the tree in the final total ordering of blocks, rather than just the blocks on the chosen path. With such a mechanism, Conflux can now obtain much higher throughput by increasing the block generation rate. For example, in one setting, Conflux generates one $4$MB block every $5$ seconds, despite that the time needed to propagate a $4$MB block to all nodes is around $100$ seconds. (In GHOST, such high block generation rate is never helpful, since most blocks would not be on the chosen path and would be wasted.)

\begin{figure*}
	\centering
	\includegraphics[width=16cm]{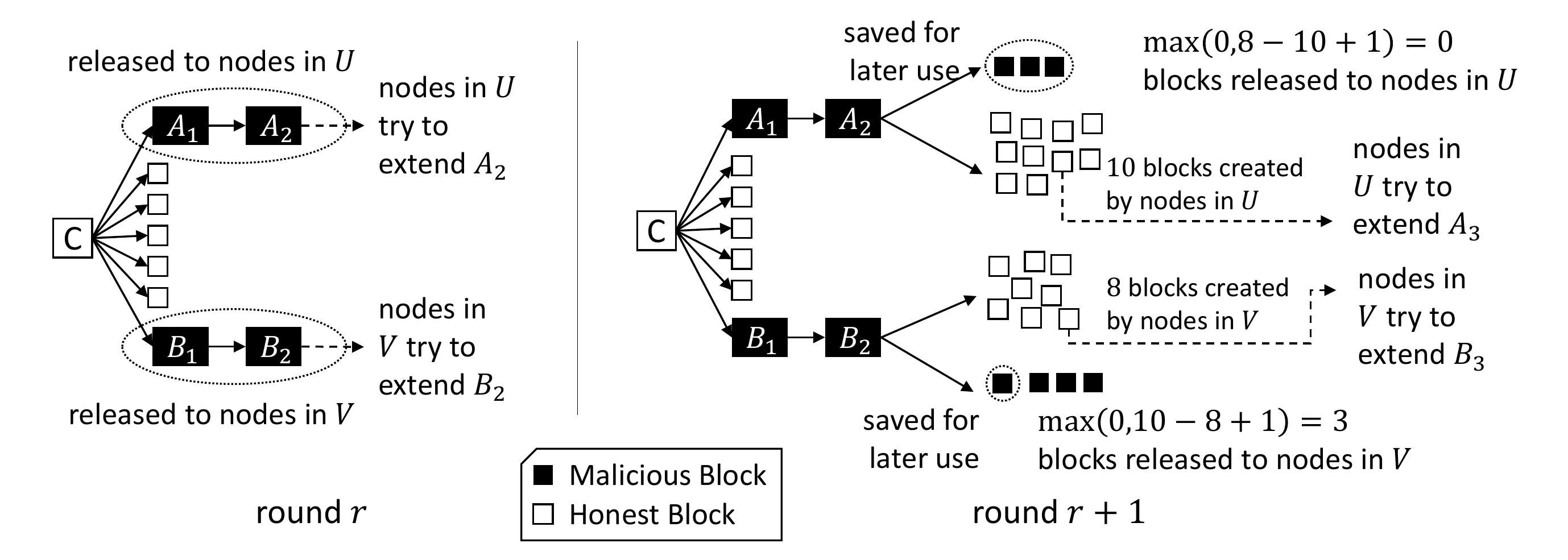} 
	\vspace{-2mm}
	\caption{Adversary's attack strategy in the balance attack.}
	\vspace{-2mm}
	\label{fig:attack}
\end{figure*}

\Paragraph{Setting for the balance attack.}
For simplicity, we will consider the following setting. We assume that Conflux proceeds in {\em rounds}, where each round corresponds to the time needed to propagate a block to all nodes. In each round, an honest node will try to mine a new block, and if successfully, will propagate the new block to all other nodes. Any blocks propagated will be received by all nodes at the end of the current round. In each round, the adversary will also mine new blocks, based on what it has seen up to the previous round. Once the honest nodes send out their newly mined blocks in a given round, we assume that adversary can observe those blocks, and then potentially send out malicious blocks in response to those. We also assume that the adversary can choose to send a malicious block to some honest nodes but not others. Since an honest node always forwards all blocks it receives, those other honest nodes will still receive those malicious blocks one round later. Conflux uses a deterministic tie-breaking rule for choosing which blocks to extend from, when there are multiple candidate blocks. We follow such a design. (Randomized tie-breaking rules will actually make the attack more effective, since randomized rules will result in honest nodes splitting their computation power.) Finally, Conflux mentions a stale block detection heuristic, which we do not consider since with the large threshold used in Conflux, it is unlikely to affect our results. 

\Paragraph{Intuition behind the balance attack.}
Before detailing the balance attack, we first provide some intuitions. The balance attack aims to maintain two forks that keep growing. Each fork will have some weight, which is roughly the number of blocks in that fork. Let $U$ be any set consisting of half of all the honest nodes, and let set $V$ contain the remaining honest nodes.
Imagine that in round $r$, nodes in $U$ and nodes in $V$ are working on the two forks, respectively. The adversary wants such property to continue to hold in round $r+1$. Let $x_A$ and $x_B$ be the number of honest blocks generated by nodes in $U$ and in $V$, respectively. The adversary will send $y_A$ malicious blocks (extending the first fork) to nodes in $U$ (but not nodes in $V$), such that $x_A+y_A > x_B$. Similarly, it sends $y_B$ malicious blocks (extending the second fork) to nodes in $V$, such that $x_B+y_B > x_A$. Doing so will keep nodes in $U$ and nodes in $V$ continuing working on each of the two forks, respectively, in round $r+1$. If the adversary do not have enough malicious blocks to follow the above step, the attack stops.

This attack is effective in Conflux but not in GHOST because with Conflux's design, high throughput is achieved by generating many block per round. For example, in one setting, Conflux generates one $4$MB block every $5$ seconds, while the time needed to propagate a $4$MB block is around $100$ seconds. This gives about $20$ blocks per round. The GHOST protocol never intends to work under such parameters. Because of this particular way of achieving high throughput, in Conflux, the expectations of $x_A$, $x_B$, $y_A$, and $y_B$ will all be much larger than those in GHOST. (The ratio of the expectations of the four quantities remain the same in Conflux as in GHOST.) In particular in GHOST, the expected number of malicious blocks generated by the adversary in one round is below $1$, which means that it is often $0$. This would render the balance attack ineffective.

\Paragraph{Attack strategy.}
Start from any  given round $r$, the following gives the detailed steps of the balance attack (Figure~\ref{fig:attack}):
\begin{enumerate}
\item In round $r$, let block $C$ be the block that all the nodes are trying to extend from. The adversary generates blocks $A_1$ extending from $C$, $A_2$ extending from $A_1$, $B_1$ extending from $C$, and $B_2$ extending from $B_1$.
\item At the end of round $r$, the adversary disseminates $A_1$ and $A_2$ to all nodes in $U$, and $B_1$ and $B_2$ to all nodes in $V$.
\item In round $r+1$, all the honest nodes in $U$ ($V$) will try to extend from $A_2$ ($B_2$). Let $x_A$ and $x_B$ be the number of (honest) blocks generated by nodes in $U$ and $V$, respectively. At the same time, the adversary will mine as many malicious blocks as it can, with odd-numbered blocks extending from $A_2$, and even-numbered blocks extending from $B_2$.
\item At the end of round $r+1$, the adversary disseminates total $\max(0,x_B-x_A+1)$ odd-numbered blocks to nodes in $U$. This will cause nodes in $U$ to see more blocks extending from $A_2$ than blocks extending from $B_2$. Hence in the next round, nodes in $U$ will all extend from some child block $A_3$ of $A_2$. Similarly, the adversary disseminates total $\max(0,x_A-x_B+1)$ even-numbered blocks to nodes in $V$, so that all nodes in $V$ will later extend from some child block $B_3$ of $B_2$. If the adversary does not have enough blocks to disseminate in this step, the attack stops.
\item In round $r+2$, all the honest nodes in $U$ ($V$) will try to extend from $A_3$ ($B_3$). In round $r+2$ and later, the adversary's strategy is similar as the above two steps, except the following. In any given round $r'$, the adversary may not use up all the blocks it generates. Those blocks can be ``saved'' for later use. For example, when the adversary needs to disseminate $\max(0,x_B-x_A+1)$ malicious blocks to nodes in $U$, those malicious blocks do not have to extend from the last block in $A_1$'s branch --- they can extend from any block in $A_1$'s branch.

\end{enumerate}

\begin{figure}
\includegraphics[width=7cm]{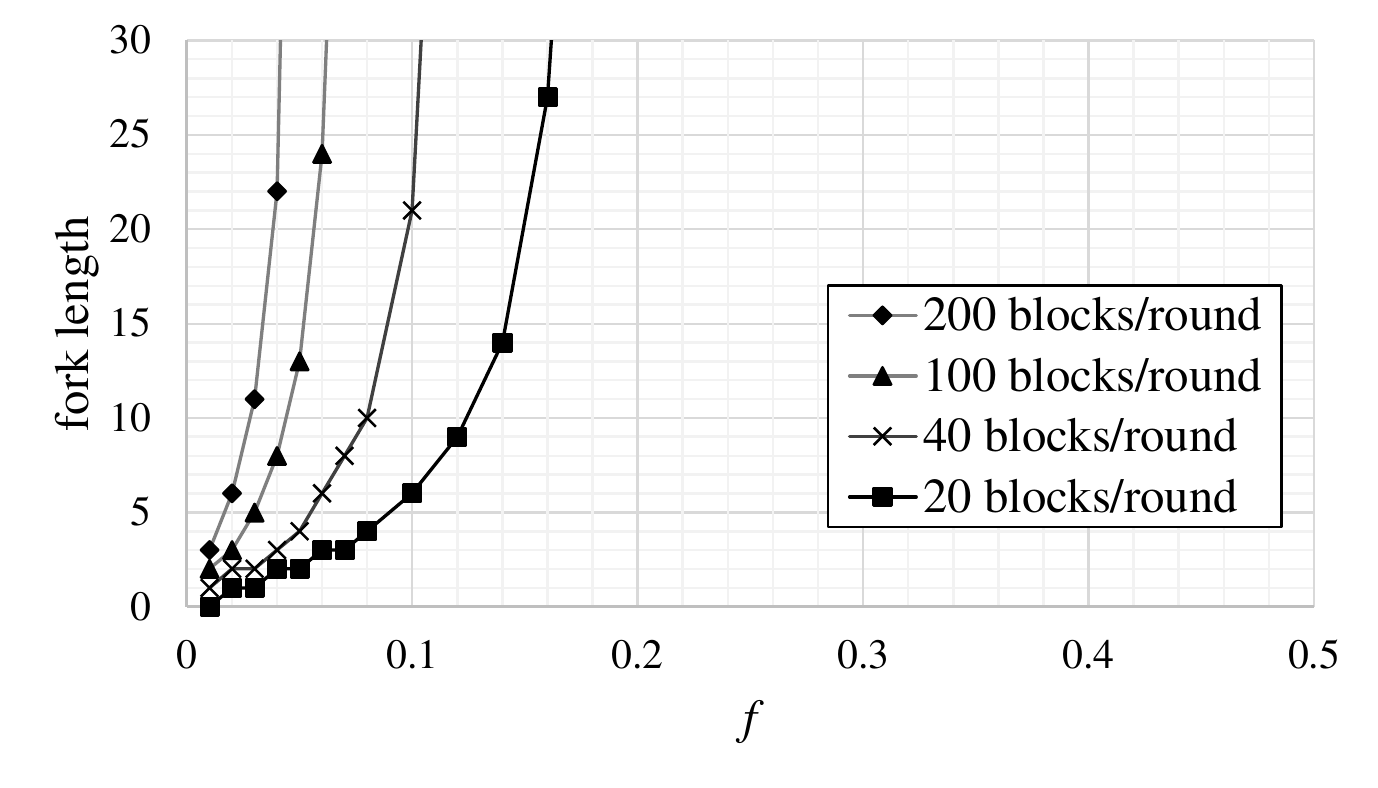}
	\vspace{-3mm}
	\caption{Fork length in Conflux when under the balance attack.}
	\vspace{-3mm}
	\label{fig:attack_sim}
\end{figure}

\Paragraph{Results.}
With the above balance attack, we use simulation to determine how long a fork the adversary can maintain. This would then correspond to the number $T$ of blocks we need to have after a block, in order for the block to get confirmed. As a reference point, Ethereum (which is also based on the GHOST protocol~\cite{ghost}) uses $T=15$. We consider the top $0.01$\% of the fork length --- this would match the $0.01$\% ``risk tolerance'' used in Conflux's evaluation~\cite{conflux}. We vary the expected number of blocks generated per round from $20$ to $200$. Out of these blocks, on expectation $f$ fraction are malicious blocks. In one of its evaluation settings, Conflux generates one $4$MB block every $5$ seconds, which would roughly correspond to $20$ blocks per round.

Figure~\ref{fig:attack_sim} presents our simulation results. It shows that with a target fork length of no more than $15$, in our simulation Conflux can only tolerate $f$ less than $0.2$ when the number of blocks per round is $20$. This $f$ deteriorates as the throughput of Conflux (i.e., number of blocks per round) increases. For example, with a target fork length of no more than $15$ and with $200$ blocks per round, in our simulation Conflux can only tolerate $f$ less than $0.1$.

\end{appendix}

\end{document}